\documentclass[11pt,a4paper,oneside]{article}
\usepackage[top=3cm, bottom=3cm, left=2.5cm, right=2.5cm]{geometry}
\usepackage[english]{babel}
\usepackage[utf8x]{inputenc}
\usepackage[T1]{fontenc}

\usepackage{amsthm,amsmath,amssymb,mathrsfs,dsfont,mathtools}
\usepackage{bm}
\usepackage{graphicx}
\usepackage[colorinlistoftodos]{todonotes}
\usepackage[colorlinks=true, allcolors=blue, hypertexnames=false]{hyperref}

\usepackage{booktabs,subcaption}
\usepackage{hyperref}
\usepackage[all]{nowidow}
\usepackage{setspace}

\usepackage[ruled,vlined]{algorithm2e}
\usepackage{tikz} % For the plot
\usepackage{lipsum}
\usepackage{adjustbox}
\usepackage{lineno} 
\usepackage{macros}
\usepackage{verbatim,amssymb,epsfig}  
\usepackage{amsmath, amsfonts, amstext,amsthm}
\usepackage{mathtools}
\usepackage{enumerate}
\usepackage{enumitem}
\usepackage{natbib}
\usepackage{graphicx}
\usepackage[title]{appendix}
\usepackage{subcaption}
\usepackage{booktabs} % for professional tables
\definecolor{revision_color}{HTML}{008080}
\usepackage{hyperref}
\usepackage{dsfont}
\hypersetup{
colorlinks = true,
allcolors = blue
}
\usepackage[font=small]{caption}
\usepackage{url} 
\usepackage{authblk}

\usepackage{multirow, makecell}

\usepackage{authblk}
\title{Marginally interpretable spatial logistic regression with bridge processes}
\date{}
\author[1]{Changwoo J. Lee \thanks{changwoo.lee@duke.edu}}
\author[1]{David B. Dunson} 
\affil[1]{Department of Statistical Science, Duke University}

\newtheorem{theorem}{Theorem}

\newtheorem{proposition}{Proposition}

\theoremstyle{definition}
\newtheorem{definition}{Definition}
\newtheorem{remark}{Remark}

\def\T{{ \mathrm{\scriptscriptstyle T} }}

\def\pibr{\normalfont{p_\textsc{br}}}
\def\pim{\normalfont{p_\textsc{m}}}

\def\bmbetam{\normalfont{\bm\beta^{\textsc{m}}}}
\def\bmbetamhat{\hat{\bm\beta}^\textsc{m}}

\begin{document}

\maketitle
\begin{abstract}
In including random effects to account for dependent observations, the odds ratio interpretation of logistic regression coefficients is changed from population-averaged to subject-specific. This is unappealing in many applications, motivating a rich literature on methods that maintain the marginal logistic regression structure without random effects, such as generalized estimating equations. However, for spatial data, random effect approaches are appealing in providing a full probabilistic characterization of the data that can be used for prediction. We propose a new class of spatial logistic regression models that maintain both population-averaged and subject-specific interpretations through a novel class of bridge processes for spatial random effects. 
These processes are shown to have appealing computational and theoretical properties, including a scale mixture of normal representation. The new methodology is illustrated with simulations and an analysis of childhood malaria prevalence data in Gambia.
\end{abstract}

\noindent%
\vfill

\section{Introduction}
Mixed-effects logistic regression with a spatial random effect serves as a canonical model for analyzing spatially indexed binary data. The generic form of the model is
\begin{equation}
     \mathrm{logit}[\pr\{Y_{ij} = 1 \mid \bfx_{ij}, u(s_i)\}] = \bfx_{ij}^\T\bm\beta + u(s_i), \quad u(\cdot) \sim \textrm{mean zero process}, \label{eq:model}
\end{equation}
where $Y_{ij}$ is the $j$th binary observation at the $i$th spatial location $s_i$ defining clusters ($i=1,\dots,n$, $j=1,\dots,N_i$), $\bfx_{ij} \in\bbR^{p}$ are covariates, and $\{u(s): s\in \scrS\}$ is a mean zero stochastic process defined on a spatial domain $\scrS$, with the most common choice being the Gaussian process. 
This generalized linear mixed effects model formulation \citep{Diggle1998-xw} provides a full characterization of the data, naturally accounting for spatial dependence, and allowing for prediction at new locations in a coherent probabilistic framework. 

When interpreting the regression coefficient $\bm\beta$, the logit link provides a familiar understanding in terms of log odds ratios. However, the inclusion of random effects changes the meaning of $\bm\beta$ from a marginal effect on the population to a conditional effect specific to the site.
Specifically, when the normal random effect $u$ is integrated out from \eqref{eq:model}, the induced model for $\logit\{\pr(Y_{ij} =1 \mid \bfx_{ij})\}$ is no longer a linear function of $\bfx_{ij}$ \citep{Zeger1988-ia}. 
There exists a rich literature on methods designed for marginally specified models, most notably generalized estimating equations \citep{Liang1986-dp}, but such approaches cannot provide probabilistic prediction at new locations. 
This creates difficulty for researchers who wish to maintain a population-averaged interpretation of $\bm\beta$ while also accounting for spatial dependence and making predictions. 
We refer the reader to \citet{Neuhaus1991-zk, Heagerty2000-nc, Hubbard2010-jf} for a detailed comparison of marginal and conditional models for clustered binary data.

For logistic models with independent random intercepts, \citet{Wang2003-ih} proposed a family of univariate bridge distributions as an alternative to normal random effects. 
In marginalizing random intercept logistic regression models over the bridge distribution, the resulting model also has a logistic form so that
$\logit\{\pr(Y_{ij} =1 \mid \bfx_{ij})\}$, where $u(s_i)$ is integrated out, is a linear function of $\bfx_{ij}$ with coefficients proportional to $\bm\beta$. 
Thus, bridge-distributed random effects allow both a conditional and marginal interpretation of regression coefficients, motivating applications in many different contexts \citep{Bandyopadhyay2010-ef,Tu2011-yc, Asar2021-sz}. However,
the bridge distribution has not been naturally extended to multivariate settings for correlated random effects.  

We propose a new class of marginally interpretable spatial logistic regression models based on a novel
spatial bridge process.
We identify the normal-scale mixture representation of the bridge distribution and propose its multivariate extension. 
In contrast to existing copula-based constructions that have bridge-distributed marginals \citep{Lin2010-sq, Li2011-kl, Parzen2011-bi,Boehm2013-xj, Swihart2014-gb}, the bridge process has appealing properties, such as transparent correlation structure and comes with significant computational benefits.  We defer all proofs to Supplementary Section~\ref{appendix:proof}.

\section{Marginally interpretable spatial logistic regression models}

\subsection{Review of bridge distribution and random intercept logistic models}

We first review bridge-distributed random intercept logistic regression models without spatial dependence. \citet{Wang2003-ih} studied a class of random-intercept distributions with density $\pibr(u;\phi)$ such that, after integrating out the random intercepts, the resulting marginal model remains logistic with its coefficients multiplied by $\phi$. 
This amounts to finding $\pibr(u; \phi)$ satisfying the integral identity
$\int_{-\infty}^\infty \mathrm{logit}^{-1}(\eta +u)\pibr(u; \phi) \rmd u = \mathrm{logit}^{-1}(\phi\eta)$ for any $\eta\in\bbR$, where $\phi$ must satisfy $\phi\in(0,1)$. The corresponding distribution, called the bridge distribution with logit link, has density 
\[
\pibr(u; \phi) = \sin (\phi \pi)/[2\pi\{\cosh (\phi u) + \cos(\phi \pi)\}], \quad u\in\bbR
\]
with parameter $\phi\in(0,1)$. 
See Figure~\ref{fig:1} (left) for a comparison with the normal distribution. 
Under the bridge-distributed random intercept logistic model, we have
\begin{align}
\logit[\pr\{Y_{ij} = 1 \mid \bfx_{ij}, u_i\}] &= \bfx_{ij}^\T\bm\beta + u_i, \quad u_i\iidsim \pibr(u; \phi) \label{eq:randint1} \\
\logit\{\pr(Y_{ij} =1 \mid \bfx_{ij})\} &= \phi\bfx_{ij}^\T\bm\beta = \bfx_{ij}^\T\bmbetam   \label{eq:randint2}
\end{align}
for $i=1,\dots,n$, $j=1,\dots,N_i$, where \eqref{eq:randint1} implies \eqref{eq:randint2} after integrating out $u_i$ so that $\bmbetam = 
\phi\bm\beta$ has an explicit marginal, population-averaged interpretation, where $\phi$ serve as an attenuation factor \citep{Gory2021-er}.  
Although the marginal probability of $Y_{ij}$ in equation \eqref{eq:randint2} depends only on  $\phi\bm\beta$, the parameter $\phi$ separately determines the variance of the random intercept. This implies that, after integrating out $u_i$, $\phi$ governs the degree of overdispersion in the cluster sums $\sum_{j=1}^{N_i} Y_{ij}$ relative to the binomial distribution \citep{Wang2004-ms}, which is formally stated in the following proposition in a simplified setting. 

\begin{proposition} 
\label{prop:varsumy}
Assume that covariates are constant within the cluster, i.e. $\bfx_{ij} \equiv \bfx_i$ for $j=1,\dots,N_i$. Then $\phi = E(\var(Y_{ij}\mid \bfx_i, u_i ))/\var(Y_{ij}\mid \bfx_i)$, the proportion of the variance of $Y_{ij}$ that is not due to the variability of $u_i$, and
\[
\textstyle{
\var\left(\sum_{j=1}^{N_i} Y_{ij} \mid \bfx_i \right) = N_i \var(Y_{ij}\mid \bfx_i)\{1 + (N_i -1)(1-\phi) \}}
\]
\end{proposition}

Thus, as $\phi$ becomes closer to 1, the within-cluster correlation induced by the random intercept disappears, and the model reduces to an ordinary logistic regression model with independent observations, i.e. $\var(\sum_{j=1}^{N_i} Y_{ij} \mid \bfx_i ) = N_i \var(Y_{ij}\mid \bfx_i)$.

\subsection{Normal scale mixture representation}

For analyzing spatial binary data, we aim to develop a spatial extension of the bridge random intercept model \eqref{eq:randint1} with appealing properties, including a dual interpretation of coefficients. 
This requires constructing a mean-zero stochastic process for the spatial random effect whose finite-dimensional distributions have bridge-distributed marginals.
We first identify a scale mixture of normal representation of the bridge distribution, which forms the basis of our construction for multivariate extensions. 

\begin{theorem}
\label{thm:unimixture}
    The bridge distribution admits the scale mixture of normal representation,
   \begin{equation}
   \label{eq:unimixture}
   \pibr(u; \phi) = \frac{\sin (\phi \pi)}{2\pi\{\cosh (\phi u) + \cos(\phi \pi)\}} =  \int_0^\infty \mathrm{N}_1(u; 0, \lambda)\pim(\lambda; \phi) \rmd \lambda,    
   \end{equation}
with the mixing variable $\lambda$ with density $\pim(\lambda; \phi)$ 
   is equal in distribution to $ 2\phi^{-2}\sum_{k=1}^{\infty}A_kB_k/k^2$, where $A_k$ is an exponential random variable with mean 1, $B_k$ is a Bernoulli (binary) random variable mean $1-\phi^2$, both independently for $k\in\bbN$.   
\end{theorem}
Our key contribution is the explicit characterization of the mixing distribution, denoted as $\lambda\sim \pim(\phi)$; the proof is similar to \citet{West1987-gt}. 
The mixing distribution has a similar form to the squared Kolmogorov distribution that serves as a normal variance mixing distribution of the logistic distribution \citep{Andrews1974-ms}, except for the presence of Bernoulli random variables $\{B_k\}$. 
More details of the mixing distribution, such as its density, can be found in Supplementary Section~\ref{appendix:multibridge}. 

\subsection{Bridge processes for logit link}

We now introduce a multivariate extension of the univariate bridge distribution, which is uniquely determined by $\phi$ and the choice of a correlation kernel $\calR$.

\begin{definition}[Bridge processes]
\label{def:bp}
Let $\calR:\scrS\times \scrS\to [-1,1]$ be a positive semidefinite kernel with $\calR(s,s) = 1$ for every $s\in \scrS$. 
We say $\{u(s) \in \bbR: s \in \scrS\}$ is a bridge process with parameter $\phi\in(0,1)$ and correlation kernel $\calR$ if every finite collection $\bfu_{1:n}=\{u(s_1),\dots,u(s_n)\}^\T$ follows the multivariate bridge distribution defined as
\begin{equation}
\label{eq:multibridgedef}
 \bfu_{1:n} \given \lambda \sim \mathrm{N}_n\left(0,  \lambda \bfR\right), \quad \lambda \sim \pim(\phi).
\end{equation}
where $\bfR$ is a $n\times n$ matrix with $(i,j)$th element $\calR(s_i,s_j)$. Here $\mathrm{N}_n(\bm\mu, \Sigma)$ denotes  $n$-dimensional multivariate normal distribution with mean $\bm\mu$ and covariance $\bm\Sigma$.
\end{definition}

\begin{figure}
    \centering
    \includegraphics[width=\linewidth]{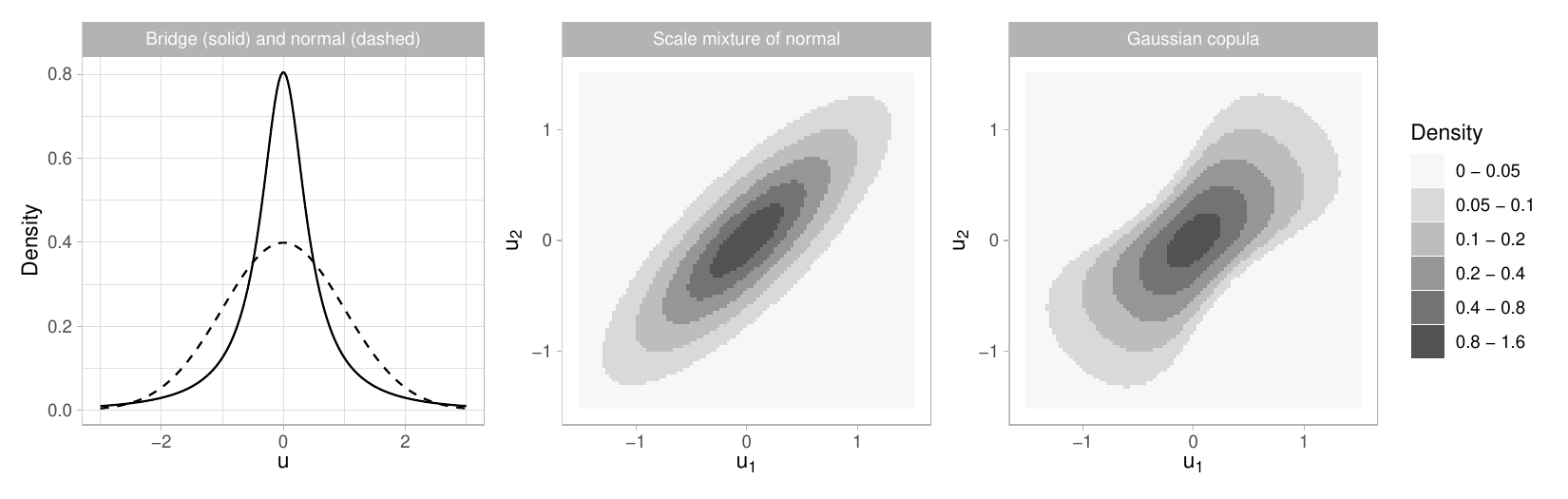}
    \caption{(Left) Density of the bridge distribution with $\phi=(1+3/\pi^2)^{-1/2}\approx 0.876$ (solid) and normal distribution (dashed), both with unit variance. (Center, Right) Two different bivariate distributions with bridge-distributed marginals with parameter $\phi$ based on the scale mixture of normal and Gaussian copula, both with correlation 0.8.}
    \label{fig:1}
\end{figure}

The requirement that $\calR(s,s)= 1$ is crucial, which ensures bridge-distributed marginals through Theorem~\ref{thm:unimixture}. Another key requirement for a stochastic process to be used for spatial prediction is Kolmogorov consistency, meaning that finite-dimensional realizations \eqref{eq:multibridgedef} must be compatible under marginalization; see \citet[][\S3.1]{Banerjee2014-yg}.
Further details on multivariate bridge distributions, including their densities, are available in Supplementary Section~\ref{appendix:multibridge}.

Existing multivariate and process extensions of the bridge distribution rely on copula-based formulations \citep{Lin2010-sq, Li2011-kl, Parzen2011-bi, Boehm2013-xj, Swihart2014-gb}. For example, with a Gaussian copula, one can define a process with bridge-distributed marginals through a nonlinear, coordinatewise transformation of a mean-zero, unit-variance Gaussian process $\zeta(\cdot)$ via $\zeta \mapsto F_{B}^{-1}(\Phi(\zeta); \phi)$, where $\Phi$ is a standard normal c.d.f. and $F_B^{-1}(\cdot; \phi)$ is the inverse c.d.f. of the bridge distribution. However, compared to the copula-based formulation, the proposed bridge process in Definition~\ref{def:bp} through a scale mixture of normals offers several practical benefits. 

First, by the scale mixture of normal construction, the realization of the proposed bridge process is elliptically symmetric \citep{Fang1990-ub}, meaning that the contours of the probability density form an ellipse/ellipsoids; see Figure~\ref{fig:1} for a comparison with a copula-based formulation. This elliptical symmetry leads to the following proposition.
\begin{proposition}
\label{prop:corr}
Let $\{u(s): s\in \scrS\}$ be a bridge process with parameter $\phi$ and kernel $\calR$. Then $\corr\{u(s),u(s')\} = \calR(s,s')$ for any $s,s' \in \scrS$. 
\end{proposition}
The Proposition~\ref{prop:corr} implies that popular choices of spatial kernels $\calR_\rho$ with parameter(s) $\rho$ retain familiar interpretations in terms of correlation between spatial random effects $u(\cdot)$, such as the range and smoothness parameters of Mat\'ern kernels.  
This is in contrast to the copula-based formulation, where the induced dependence structure via a latent process is obscure due to the complex nonlinear transformations involved, i.e. $\corr\{F_{B}^{-1}(\Phi(\zeta(s));\phi), F_{B}^{-1}(\Phi(\zeta(s'));\phi)\}\neq \corr\{\zeta(s), \zeta(s')\}$ for a latent Gaussian process $\zeta(\cdot)$.

There are elliptical distributions that cannot be represented as scale mixtures of normals \citep{Gomez-Sanchez-Manzano2006-rb}. This raises the natural question of whether there exists a bridge process that does not follow Definition~\ref{def:bp} but does have the elliptical symmetry property. 
The following proposition shows that the answer is no, further supporting our proposed formulation.

\begin{proposition}
\label{prop:characterization}
    The bridge process is the only process that has (i) Kolmogorov consistency, (ii) bridge-distributed marginals, and (iii) elliptically symmetric realizations.
\end{proposition}

Next, because the bridge process specifies a transparent correlation structure for the spatial random effect, it can be incorporated naturally into the standard model-based geostatistics framework, which consists of specifying a data model, a process model, and a parameter model \citep{Berliner1996-nd,Wikle2019-tq}. In this framework, the process model introduces dependence structure on the linear predictor scale, typically focusing on second-order information to characterize spatial dependence \citep{Gelfand2016-zm}. For binary data with a logit link, \citet{Diggle2016-lc} refer to this formulation as the standard geostatistical model for prevalence data. 
Finally, the normal scale mixture representation of the bridge process provides significant benefits for posterior computation together with P\'olya-Gamma augmentation \citep{Polson2013-gb}, which will be elaborated in detail in Section 3.2.

\subsection{Spatial random effect modeling with bridge processes}

We introduce a marginally interpretable spatial logistic regression model with a bridge process,
\begin{equation}
     \mathrm{logit}[\pr\{Y_{ij} = 1 \mid \bfx_{ij}, u(s_i)\}] = \bfx_{ij}^\T\bm\beta + u(s_i), \quad u(\cdot) \sim \mathrm{bridge\,process}(\phi, \calR), \label{eq:proposedmodel}
\end{equation}
for $i=1,\dots,n$ and $j=1,\dots,N_i$. As in the random intercept model, the proposed model \eqref{eq:proposedmodel} induces a marginal logistic regression model \eqref{eq:randint2}, where $\bm\beta$ and $\phi \bm\beta$ carry site-specific and population-averaged interpretations, respectively. 
To provide a clearer picture of the induced dependence structure between binary outcomes, we introduce an equivalent representation of \eqref{eq:proposedmodel} based on the thresholding of latent variables. 
\begin{proposition}
\label{prop:latentrep}
    The spatial logistic model with bridge process \eqref{eq:proposedmodel} is equivalent to
    \begin{equation}
    \label{eq:latentrep}
        Y_{ij} = 1(Z_{ij}>0), \quad Z_{ij} = \bfx_{ij}^\T\bm\beta + u(s_i) + \epsilon_{ij}, \quad u(\cdot)\sim \mathrm{bridge\,process}(\phi, \calR),
    \end{equation}
    where $\epsilon_{ij}$ follows independent standard logistic distributions. Furthermore, the marginal distribution of $u(s_i) + \epsilon_{ij}$ is a mean zero logistic distribution with scale parameter $\phi^{-1}$.
\end{proposition}
By multiplying $\phi$ at both sides of the second display of \eqref{eq:latentrep}, it is evident that $\phi \bm\beta$ carries population-averaged interpretation since  $
\phi u(s_i) + \phi\epsilon_{ij}$ follows the standard logistic distribution.

\section{Posterior inference}

\subsection{Empirical Bayes approach for attenuation factor}

The proposed marginally interpretable spatial logistic model \eqref{eq:proposedmodel} has parameters $(\bm\beta, \phi, \rho)$, where $\rho$ denotes a correlation kernel parameter. The complete data likelihood, including all latent variables, is given by
\begin{align}
  \calL_{\mathrm{comp}}(\bm\beta, \phi, \rho, \bfu_{1:n}, \lambda) = \mathrm{N}_n(\bfu_{1:n}; \bm{0}, \lambda \bfR_\rho)\pim(\lambda; \phi) \prod_{i=1}^n \left\{\prod_{j=1}^{N_i}\frac{\exp(\bfx_{ij}^\T\bm\beta + u(s_i))^{y_{ij}}}{1+\exp(\bfx_{ij}^\T\bm\beta + u(s_i))}\right\}  \label{eq:complete}
\end{align}
where $\mathrm{N}_n(\bfu_{1:n}; \bm{0}, \lambda \bfR_\rho)$ is a $n$-dimensional mean zero covariance $\lambda \bfR_\rho$ multivariate normal density evaluated at $\bfu_{1:n}$. The expression \eqref{eq:complete} can be viewed as a usual spatial logistic regression with a Gaussian process random effect with further hierarchical formulation on $\lambda$, where mixing density $\pim(\lambda;\phi)$ is now part of the complete data likelihood. The likelihood $\calL(\bm\beta, \phi, \rho) = \int_{0}^\infty \int_{\bbR^n}\calL_{\mathrm{comp}}(\bm\beta, \phi, \rho, \bfu_{1:n}, \lambda) \rmd\bfu_{1:n}\rmd\lambda$ involves high-dimensional integrations which introduces significant computational challenges. We therefore adopt a Bayesian framework, employing Markov chain Monte Carlo for posterior inference on the key parameters of interest, site-specific $\bm\beta$ and population-averaged effects $\bmbetam = \phi\bm\beta$, as well as prediction of probability of binary response at new spatial locations. 

However, when observed data is purely spatial and does not have replicates (i.e., with only a single realization $\bfu_{1:n}$ from the bridge process), fully Bayesian inference for $(\bm\beta, \phi, \rho)$ faces several challenges. 
It is well known that in spatial generalized linear models with Gaussian process random effects, the joint estimation of marginal variance and spatial dependence parameters is difficult \citep{Zhang2004-vn}. This issue becomes even more pronounced in our setting, as $\phi$ acts as a hyperparameter in $\lambda \sim \pim(\phi)$ with a deeper hierarchical formulation. While $\phi$ solely controls the marginal variance of the latent bridge process, in the absence of replicates, data provide limited information about $\phi$, resulting in weak identifiability of $\phi$ where the posterior of $\phi$ is highly influenced by its prior \citep{Garrett2000-py}.

Although a fully Bayesian approach is possible, we suggest an empirical Bayes procedure that first finds a point estimate $\hat\phi$ and then computes the posterior of $(\bm\beta, \rho)$ given $\phi = \hat\phi$. Instead of using marginal maximum likelihood for determining $\hat{\phi}$, which involves high-dimensional integration and is sensitive to model misspecification, we estimate $\hat{\phi}$ by maximizing a composite likelihood based on univariate and bivariate marginals, providing greater robustness \citep{Varin2011-ep}. Specifically, assuming $N_i \ge 2$ for all $i$, we adopt the two-stage composite likelihood method of \citet{Zhao2005-qf}. In the first stage, leveraging the connection with the marginal logistic regression, we first get a maximum likelihood estimate (MLE) $\bmbetamhat$ from the marginal logistic model \eqref{eq:randint2} ignoring the cluster structure. In the second stage, given $\bmbetamhat$, we set $\hat\phi = \argmax_{\phi\in(0,1)} \calL_{\mathrm{pair}}(\phi; \bmbetam = \bmbetamhat)$ using a pairwise composite likelihood constructed from within-cluster pairs, which is
\begin{equation}
\label{eq:pair}
    \calL_{\mathrm{pair}}(\phi; \bmbetam) = \prod_{i=1}^n\prod_{j<j'}\int_0^\infty \frac{\exp(\phi^{-1}\bfx_{ij}^\T\bmbetam + u(s_i))^{y_{ij}}}{1+\exp(\phi^{-1}\bfx_{ij}^\T\bmbetam + u(s_i))}\frac{\exp(\phi^{-1}\bfx_{ij'}^\T\bmbetam + u(s_i))^{y_{ij'}}}{1+\exp(\phi^{-1}\bfx_{ij'}^\T\bmbetam + u(s_i))} \pibr(u(s_i); \phi) \rmd u(s_i).
\end{equation}
It follows from \citet{Zhao2005-qf} that, assuming conditional \eqref{eq:randint1} and marginal models \eqref{eq:randint2} are correctly specified but without assumptions on the dependence structure of spatial random effect $u(\cdot)$, the resulting empirical Bayes estimator $\hat\phi$ is consistent under increasing-domain asymptotics. This holds because $\bmbetamhat$ is a solution of the unbiased estimating equation arising from the marginal logistic model, and the pairwise likelihood utilizes the bridge-distributed marginal of the random effect but not the spatial dependence information.

\subsection{Posterior computation}
We describe a two-stage posterior computation procedure in Algorithm~\ref{alg:gibbs} based on the empirical Bayes estimate of $\phi$; a description of a fully Bayesian approach is available in Supplementary Section~\ref{appendix:post}.  
For simplicity, we focus on prior $p(\bm\beta,\rho) = p(\bm\beta)p(\rho)$ with a normal prior for $\bm\beta$ and some proper prior on $\rho$. %, 

\begin{algorithm}
%\small
\caption{Posterior computation procedure with empirical Bayes estimator for $\phi$.}
\hspace{-3mm}\textbf{Stage 1}. Find  empirical Bayes estimate $\hat\phi$ by following steps:\\
{[1]} Compute maximum likelihood estimate $\bmbetamhat$ from simple logistic regression \eqref{eq:randint2}.\\
{[2]} Set $\hat\phi = \argmax_{\phi\in(0,1)}\calL_{\mathrm{pair}}(\phi; \bmbetam = \bmbetamhat)$ in eq. \eqref{eq:pair} \\
\hspace{-3mm}\textbf{Stage 2}. Compute the posterior by MCMC by repeating the following steps: \\
\textbf{Repeat}:\\
 {[1]} Sample $\bm\beta \sim [\bm\beta \mid \lambda^{\mathrm{(old)}}, \rho^{\mathrm{(old)}},  \bm\omega^{\mathrm{(old)}}]$ from multivariate normal. \\
 {[2]} Sample $\rho \sim [\rho\mid\lambda^{\mathrm{(old)}}, \bm\beta,  \bm\omega^{\mathrm{(old)}}]$ using adaptive Metropolis-Hastings.\\
 {[3]} Sample $\lambda \sim [\lambda\mid\rho, \bm\beta, \bm\omega^{\mathrm{(old)}}]$ using independent Metropolis-Hastings\\
 {[4]} Sample $\bfu_{1:n}\sim [\bfu_{1:n}\mid \lambda, \rho, \bm\beta,\bm\omega^{\mathrm{(old)}}]$ from multivariate normal. \\
 {[5]} Sample $\omega_{ij} \sim [\omega_{ij} \mid \bm\beta, \bfu_{1:n}]$ from P\'olya-Gamma, independently for all $i,j$.\\
\label{alg:gibbs}
\end{algorithm}

Stage 2 of Algorithm~\ref{alg:gibbs} describes a partially collapsed Gibbs sampler \citep{Van-Dyk2008-dr}, and we outline key strategies involved. Writing $\bfY=\{y_{ij}\}$, $\bfX = \{\bfx_{ij}\}$,  augmented data $\bm\omega = \{\omega_{ij}\}$ and $p_{\textsc{pg}}$ as the density of the P\'olya-Gamma$(1,0)$ distribution, the complete data model via P\'olya-Gamma augmentation \citep{Polson2013-gb} and prior becomes 
\begin{align}
  p(\bfY, \bm\omega \mid \bfX, \bm\beta, \bfu_{1:n}) &= \prod_{i=1}^n\prod_{j=1}^{N_i}\exp\Big[\Big(y_{ij}-\frac{1}{2}\Big)\{\bfx_{ij}^\T\bm\beta +u(s_i)\}- \frac{\omega_{ij}}{2}\{\bfx_{ij}^\T\bm\beta +u(s_i)\}^2\Big]\frac{p_{\textsc{pg}}(\omega_{ij})}{2} \label{eq:expandedmodel1}\\
 (\bfu_{1:n}\mid \lambda) &\sim \mathrm{N}_n(0, \lambda \bfR), \quad \lambda \sim p_{\textsc{m}(\hat\phi)}, \quad  p(\bm\beta, \rho) = \mathrm{N}_p(\bm\beta;\bm\mu_\beta, \bm\Sigma_\beta)\times p(\rho) \label{eq:expandedmodel2}
\end{align}
where \eqref{eq:expandedmodel1} satisfies $\int p(\bfY,\bm\omega\mid \bfX,\bm\beta,\bfu_{1:n}) \rmd \bm\omega = p(\bfY \mid \bfX,\bm\beta,\bfu_{1:n})$ corresponding to \eqref{eq:model}. The log-likelihood conditional on $\bm\omega$ is a quadratic function in both $\bm\beta$ and $\bfu_{1:n}$. 
Since the spatial random effect $\bfu_{1:n}$ is normal conditional on mixing variable $\lambda$, $\bfu_{1:n}$ can be analytically integrated out, leading to partial collapsing in steps 1 to 3 in Stage 2.  
This leads to improved mixing, especially for the intercept term in $\bm\beta$ that is often highly correlated with $\bfu_{1:n}$. Conditionally on P\'olya-Gamma variables $\bm\omega$, we update $\bm\beta$ and $\bfu_{1:n}$ from multivariate normal leveraging conditional conjugacy. When updating $\lambda$, we use Metropolis-Hastings with an independent proposal from the prior. 
The remaining steps and posterior prediction procedure are straightforward, and we defer the detailed derivations to Supplementary Section~\ref{appendix:post}.

\subsection{Scalable computation with low-rank dependence structure}
\label{subsec:scalable}
The Algorithm~\ref{alg:gibbs} involves several inversions and determinant calculations of the $n\times n$ matrices, creating a computational bottleneck when the number of spatial locations $n$ is large. Specifically, step 1 of Stage 2 involves the inversion of an $n\times n$ matrix, and steps 2, 3, and 4 of Stage 2 involve the evaluation of the density and sampling from the multivariate normal distribution of dimension $n$. Without special structures in the corresponding covariance or precision matrices, the computation becomes prohibitive as the number of spatial locations $n$ increases.

Due to the normal mixture representation, several existing computational strategies for Gaussian processes can be easily integrated into bridge processes. One way is to introduce a low-rank structure on the correlation kernel, following the strategy of \citet{Finley2009-bc} for Gaussian processes. Given a positive definite correlation kernel $\calR$, denote $\bfR_{qq} = [\calR(\tilde{s}_k,\tilde{s}_{k'})]_{k,k'=1}^q$ as a $q\times q$ matrix formed from $q$ knot locations $\tilde{s}_k\in \scrS$, $k=1,\dots,q$. Then, we consider a correlation kernel $\tilde{\calR}$ defined as   
\begin{equation}    
\label{eq:corr_lowrank}
\tilde{\calR}(s,s') = \bfr(s)^\T \bfR_{qq}^{-1}\bfr(s') + 1(s=s') \{1- \bfr(s)^\T \bfR_{qq}^{-1}\bfr(s')\},
\end{equation}
where $r(s) = [\calR(s,\tilde{s}_k)]_{k=1}^q$ is a $q\times 1$ vector and the term $1(s=s') \{1- \bfr(s)^\T \bfR_{qq}^{-1}\bfr(s')\}$ in \eqref{eq:corr_lowrank} ensures that $\tilde{\calR}(s,s) = 1$. Using the Woodbury matrix identity and the matrix determinant lemma, the computation related to the inversion and calculation of the determinant of $n\times n$ matrices can be reduced to those of $q\times q$ matrices, which gives huge computational benefits when $q \ll n$; see Supplementary Section~\ref{appendix:post} for details. In light of Proposition~\ref{prop:characterization}, the seamless incorporation of scalable Gaussian process methods through a hierarchical formulation is a distinctive feature of bridge processes.

One key requirement is that the new correlation kernel $\tilde{R}$ must maintain a unit diagonal $\tilde{\calR}(s,s) \equiv 1$. Other possible low-rank constructions include full-scale approximation \citep{Sang2012-vj}; however, ordinary predictive processes \citep{Banerjee2008-iw} and nearest-neighbor Gaussian processes \citep{Datta2016-dx} are not directly applicable because they induce heterogeneous marginal variances.

\section{Simulation studies}
\subsection{Comparison with existing approaches}
\label{subsec:sim1}

First, we conduct a simulation study to compare the proposed approach with existing methods for both marginal and conditional inference. For the proposed bridge process random effect, we consider both the empirical Bayes approach (BrP) and a fully Bayesian approach that introduces a prior on~$\phi$ (BrP-FB). 
As a marginal effect comparison, we consider the spatial generalized estimating equation (SpGEE) method based on pairwise log-odds ratios \citep{Cattelan2018-pp}. 
As a conditional effect comparison, we consider a fully Bayesian spatial logistic model with Gaussian process (GP) random effects.   

We generate data based on the spatial logistic model \eqref{eq:model} under two different processes with bridge-distributed marginals, one with the proposed model and another based on the Gaussian copula process with bridge marginals. 
This ensures both population-averaged effects $\bmbetam = \phi \bm\beta$ and site-specific effects $\bm\beta$ are well-defined. We consider attenuation factors $\phi \in \{0.7, 0.9\}$, where the choice $\phi = 0.9$ is motivated by the estimate $\hat\phi = 0.895$ obtained in the Gambia case study (Section 5).
We choose spatial locations uniformly at random from unit square domain $\scrS = [0,1]^2$ to decide training and test locations with sizes  $(n_{\mathrm{train}},n_{\mathrm{test}})=(200,50)$, where $50$ test locations are held out for assessing predictive performance. 
We use Mat\'ern correlation kernel $\calR_M(s,s') = (1+\|s-s'\|_2/\rho)\exp(-\|s-s'\|_2/\rho)$ with known smoothness $1.5$ and unknown range parameter $\rho$, where true $\rho$ are set as $\rho = 0.05$ or $\rho=0.1$. The same correlation kernel is applied to the latent Gaussian process $\zeta(\cdot)$ in the Gaussian copula. 
We set $N_i = 10$ for all locations and set $p=2$ including the intercept, where the non-intercept covariate $\{x_{ij}\}$ is generated from an independent standard normal distribution. The true fixed-effect coefficient is set as $\bm\beta = (\beta_0,\beta_1)^\T = (0,1)^\T$. This data generation process is repeated 200 times.

We describe details of prior specifications. Following \citet{Gelman2008-hm}, we choose Cauchy priors for $\bm\beta$, specifically location 0 scale 10 for $\beta_0$ and location 0 scale 1.25 for $\beta_1$. For the prior of spatial range $\rho$ in Mat\'ern correlation kernel with known smoothness 1.5, we assign a uniform prior $\rho\sim \mathrm{Unif}(0.001,0.3)$ for both bridge and Gaussian process random effect models, which corresponds to an effective range approximately between 0.00475 and 1.425 in unit square domain. For fully Bayesian approaches with prior on the random effect standard deviation $\sigma_u$, we choose a weakly informative half-Cauchy prior \citep{Gelman2006-de}, namely $p(\sigma_u) = 2/\{\pi(1+\sigma_u^2)\}$ for Gaussian process random effects. For the bridge process models where $\sigma_u = 3^{-1/2}\pi(\phi^{-2}-1)^{1/2}$, this becomes $p(\phi) = (12)^{1/2}/[\{\pi^2 - (\pi^2-3)\phi^2\}(1-\phi^2)^{1/2}]$ for $\phi\in(0,1)$, corresponding to prior mean approximately 0.753. 

Next, we describe inference settings. For the spatial generalized estimating equation method of \citet{Cattelan2018-pp} with the same notations therein, we used the number of bins $B = 13$, radius $h=0.05$, and $d_{\mathrm{max}} = 0.3$ to obtain the empirical spatial lorelogram, and optimize the parameters $\alpha_2$ (sill) and $\alpha_3$ (range) of Mat\'ern kernel with smoothness 1.5 without nugget. For the bridge and Gaussian process random effect models, we run 11,000 iterations, with the first 1,000 samples discarded as burn-in and 1,000 samples are saved with 10 thin-in rates. The algorithms are all implemented in \texttt{R}, and the running time of Markov chain Monte Carlo is about 38 mins for the bridge process random effect model and about 20 mins for the Gaussian process random effect model under the Intel(R) Xeon(R) Gold 6336Y 2.40GHz CPU environment.

\begin{table}

\footnotesize
\caption{Comparison of estimated population-averaged effects based on 200 simulations in terms of bias, root mean squared error (RMSE), average length of 95\% confidence or credible interval (CI$_{.95}$), and coverage probabilities at nominal level 0.95 (Cover).}
    \label{table:simmar}
    \centering
    \begin{tabular}{cc c c c c c c c c}
    \toprule
     \multicolumn{2}{c}{Population-averaged ($\hat\beta_1^\textsc{m}$)} & \multicolumn{4}{c}{Data from bridge process} & \multicolumn{4}{c}{Data from Gaussian copula}\\
      \cmidrule(lr){3-6} \cmidrule(lr){7-10}
     \cmidrule(lr){3-4}\cmidrule(lr){5-6} \cmidrule(lr){7-8}\cmidrule(lr){9-10}
     Setting & Method & Bias & RMSE & CI$_{.95}$ & Cover & Bias & RMSE & CI$_{.95}$ & Cover\\ 
     \midrule
    \multirow{3}{*}{\makecell{ $\phi = 0.7$,\\ $\rho = 0.05$ }} &  SpGEE & 0.015 & 0.176 & 0.233 & 40.0\% & 0.008 & 0.063 & 0.241 & 94.5\% \\
    & BrP & 0.008 & 0.174 & 0.189 & 33.5\% & 0.005 & 0.064 & 0.190 & 86.5\%  \\
    & BrP-FB & 0.005 & 0.128 & 0.532 & 96.5\% & 0.006 & 0.051 & 0.568 & 100.0\%  \\
    \midrule
    \multirow{3}{*}{\makecell{ $\phi = 0.7$,\\ $\rho = 0.1$ }}& SpGEE & 0.033 & 0.184 & 0.275 & 48.5\% & 0.026 & 0.084 & 0.293 & 92.0\% \\
        & BrP & 0.030 & 0.182 & 0.193 & 31.0\% & 0.021 & 0.084 & 0.192 & 75.5\%  \\
    & BrP-FB & -0.009 & 0.140 & 0.557 & 95.0\% & -0.021 & 0.069 & 0.608 & 100.0\%  \\
         \midrule
 \multirow{3}{*}{\makecell{ $\phi = 0.9$,\\ $\rho = 0.05$ }}& SpGEE & -0.033 & 0.151 & 0.230 & 70.5\% & 0.002 & 0.058 & 0.228 & 95.0\% \\
    & BrP & -0.038 & 0.151 & 0.212 & 69.5\% & -0.003 & 0.057 & 0.218 & 94.0\%  \\
    & BrP-FB & -0.087 & 0.139 & 0.429 & 95.5\% & -0.099 & 0.114 & 0.416 & 95.0\% \\
    \midrule
 \multirow{3}{*}{\makecell{ $\phi = 0.9$,\\ $\rho = 0.1$ }} &  SpGEE & -0.030 & 0.148 & 0.252 & 74.0\% & 0.010 & 0.069 & 0.244 & 90.5\%\\
    & BrP & -0.036 & 0.151 & 0.212 & 69.0\% & 0.005 & 0.068 & 0.219 & 88.0\%  \\
    & BrP-FB & -0.105 & 0.156 & 0.464 & 94.0\% & -0.117 & 0.137 & 0.470 & 94.5\% \\
    \bottomrule
    \end{tabular}
    \vspace{-3mm}
\begin{flushleft} 
{\footnotesize SpGEE, spatial generalized estimating equation method of \citet{Cattelan2018-pp}; BrP, proposed bridge process model; BrP-FB, bridge process model with fully Bayesian approach on $\phi$.}
\end{flushleft}
\end{table}

\begin{table}
\footnotesize
\caption{Comparison of estimated site-specific effects based on 200 simulations in terms of bias, root mean squared error (RMSE), average length of 95\% confidence or credible interval (CI$_{.95}$), and coverage probabilities at nominal level 0.95 (Cover).}
    \label{table:simcond}
    \centering
    \begin{tabular}{cc c c c c c c c c}
    \toprule
       \multicolumn{2}{c}{Site-specific ($\hat\beta_1$)}  & \multicolumn{4}{c}{Data from bridge process} & \multicolumn{4}{c}{Data from Gaussian copula}\\
      \cmidrule(lr){3-6} \cmidrule(lr){7-10}
     \cmidrule(lr){3-4}\cmidrule(lr){5-6} \cmidrule(lr){7-8}\cmidrule(lr){9-10}
      Setting & Method & Bias & RMSE & CI$_{.95}$ & Cover & Bias & RMSE & CI$_{.95}$ & Cover\\ 
     \midrule
    \multirow{3}{*}{\makecell{ $\phi = 0.7$,\\ $\rho = 0.05$ }} &  GP & 0.002 & 0.068 & 0.275 & 95.5\% & -0.001 & 0.063 & 0.270 & 98.5\% \\
    & BrP  & 0.003 & 0.068 & 0.275 & 96.5\% & -0.001& 0.063 & 0.270 & 98.5\%  \\
    & BrP-FB & 0.004 & 0.068 & 0.281 & 96.5\% & 0.001 & 0.063 & 0.275 & 98.5\%  \\
    \midrule
 \multirow{3}{*}{\makecell{ $\phi = 0.7$,\\ $\rho = 0.1$ }} & GP & -0.001 & 0.074 & 0.271 & 93.5\% & -0.002 & 0.065 & 0.267 & 97.0\% \\
        & BrP & -0.002 & 0.074 & 0.271 & 94.0\% & -0.001 & 0.066 & 0.266 & 97.5\%  \\
    & BrP-FB & 0.001 & 0.074 & 0.279 & 94.5\% &  -0.002 & 0.065 & 0.276 & 97.0\% \\
         \midrule
 \multirow{3}{*}{\makecell{ $\phi = 0.9$,\\ $\rho = 0.05$ }} & GP & 0.001 & 0.064 & 0.251 & 96.0\% & -0.003 & 0.060 & 0.243 & 95.0\%\\
    & BrP  & 0.000 & 0.063 & 0.250 & 95.5\% & -0.005 & 0.061 & 0.242 & 95.5\%  \\
    & BrP (FB) & 0.001 & 0.063 & 0.252 & 96.0\% & -0.003 & 0.060 & 0.245 & 96.0\% \\
    \midrule
 \multirow{3}{*}{\makecell{ $\phi = 0.9$,\\ $\rho = 0.1$ }} &  GP & 0.000 & 0.066 & 0.249 & 94.5\%& -0.001 & 0.059 & 0.241 & 96.0\% \\
    & BrP  & -0.001 & 0.066 & 0.248 & 94.5\% & -0.003 & 0.059 & 0.241 & 95.0\%  \\
    & BrP-FB & 0.001 & 0.066 & 0.254 & 94.5\% & 0.000 & 0.059 & 0.246 & 96.5\% \\
    \bottomrule
    \end{tabular}
    \vspace{-3mm}
\begin{flushleft} 
{\footnotesize GP, Gaussian process random effect model; BrP, proposed bridge process model; BrP-FB, bridge process model with fully Bayesian approach on $\phi$.}
\end{flushleft}
\end{table}

The simulation results are summarized in terms of the estimated population-averaged effects $\hat\beta_1^\textsc{m}$ in Table~\ref{table:simmar} and the site-specific effects $\hat\beta_1$ in Table~\ref{table:simcond}. For summaries of the intercept terms, see Supplementary Section~\ref{appendix:simdetails}. We first remark that when the data are generated from bridge processes, the root mean squared error (RMSE) of $\hat\beta_1^\textsc{m}$ in Table~\ref{table:simmar} is substantially larger than that under the Gaussian copula, whereas the RMSEs of the conditional effect $\hat\beta_1$ in Table~\ref{table:simcond} are comparable across the two data generating mechanisms. This suggests that there is a large variability of population-averaged effect $\hat\beta_1^\textsc{m}$ under the bridge process data-generating scheme, compared to the Gaussian copula data-generating scheme.

In terms of the estimated population-averaged effects $\hat\beta_1^\textsc{m}$, the SpGEE and BrP produce generally similar results, whereas BrP-FB shows substantial differences. Notably, BrP-FB exhibits substantial bias when $\phi = 0.9$ for both data generation scenarios, and when $\phi = 0.7$, the bias becomes confounded with the influence of the prior $p(\phi)$. In contrast, both SpGEE and BrP show relatively robust results across different $(\phi, \rho)$ and data settings, highlighting the advantage of the empirical Bayes approach to find $\hat\phi$ based on composite likelihood. Regarding confidence and credible interval, BrP-FB produces much wider intervals than both SpGEE and BrP. Although BrP-FB yields the desired coverage when the data are generated from a bridge process random effects, it also results in 100\% coverage under Gaussian copula data, suggesting that uncertainty quantification with BrP-FB is highly sensitive to model misspecification. 

For the estimated site-specific effects $\hat\beta_1$, all methods GP, BrP, and BrP-FB lead to highly similar results. While there are slight overcoverages under the Gaussian copula data generation setting with $\phi = 0.7$, generally there are almost no biases and all three methods achieve the desired coverage level. Table~\ref{table:1pred} summarizes predictive performance based on held-out data of size $n_{\mathrm{test}} = 50$, $N_{\mathrm{test}} = 500$. Across three methods GP, BrP, and BrP-FB, we compare test log-likelihood averaged over $N_{\mathrm{test}}=500$ binary data and area under the receiver operating characteristic curve based on test data (test AUC) using the posterior predictive mean. The results suggest that the GP, BrP, and BrP-FB models produce virtually identical predictions. In summary, the results demonstrate that BrP with an empirical Bayes estimate of $\phi$ provides a marginal effect estimate comparable to SpGEE, and at the same time yield site-specific estimates and predictions that are virtually identical to the Gaussian process random effects model.

\color{black}

\begin{table}
\footnotesize
\caption{
Comparison of predictive performance (both higher the better) across three methods based on held-out binary data, averaged across 200 replicates.}
    \label{table:1pred}
\centering
\begin{tabular}{cc c c c c}
\toprule
& & \multicolumn{2}{c}{Data from bridge process} & \multicolumn{2}{c}{Data from Gaussian copula} \\
\cmidrule(lr){3-4}\cmidrule(lr){5-6}
Setting & Method & Test loglik $\times 10^2$ & Test AUC $\times 10^2$ & Test loglik $\times 10^2$ & Test AUC $\times 10^2$ \\
\midrule

\multirow{3}{*}{\makecell{$\phi = 0.7$,\\ $\rho = 0.05$}}
  & GP        & -53.069  & 80.137 & -52.572 & 80.861 \\
  & BrP       & -53.076 & 80.129 & -52.571 & 80.863 \\
  & BrP-FB  & -53.073 & 80.130 & -52.564 & 80.866 \\
\midrule
\multirow{3}{*}{\makecell{$\phi = 0.7$,\\ $\rho = 0.1$}}
  & GP        & -49.963  & 81.932  & -49.103 & 83.037 \\
  & BrP       & -49.964 & 81.929 & -49.108 &  83.034 \\
  & BrP-FB  & -49.963 & 81.931  & -49.108 & 83.033 \\
\midrule
\multirow{3}{*}{\makecell{$\phi = 0.9$,\\ $\rho = 0.05$}}
  & GP        & -57.410  &  76.437 & -58.680 & 75.339 \\
  & BrP       & -57.415  & 76.433 & -58.680 & 75.339  \\
  & BrP-FB  & -57.408 & 76.440 & -58.671 & 75.351  \\
\midrule
\multirow{3}{*}{\makecell{$\phi = 0.9$,\\ $\rho = 0.1$}}
  & GP        & -55.445 & 77.664 & -57.224 & 76.567 \\
  & BrP       & -55.445 & 77.667 & -57.229 & 76.563 \\
  & BrP-FB  & -55.446 & 77.663 & -57.223 & 76.566 \\
\bottomrule
\end{tabular}
    \vspace{-3mm}
\begin{flushleft} 
{\footnotesize GP, Gaussian process random effect model; BrP, proposed bridge process model; BrP-FB, bridge process model with fully Bayesian approach on $\phi$. The Monte Carlo standard errors of Test loglik $\times 10^2$ and Test AUC $\times 10^2$ are all less than 0.5 for $\rho=0.05$ and all less than 0.7 for $\rho = 0.1$. 
}
\end{flushleft}
\end{table}

\subsection{Scalability analysis with low-rank dependence structure}

Next, we conduct another simulation study to analyze the benefits of scalable computing strategies described in Section~\ref{subsec:scalable} when the number of spatial locations $n$ becomes large. Based on the proposed bridge process random effects model, we compare how the low-rank structure of the correlation kernel affects the prediction and sampling efficiency of the parameters. 

We generate data based on the bridge process random effects model \eqref{eq:proposedmodel}. We choose spatial locations uniformly at random from the square domain $\scrS = [0,2]^2$ to decide training and test locations with sizes $(n_{\mathrm{train}},n_{\mathrm{test}})\in\{(200,50),(800,200)\}$. Similarly to previous simulation settings, we set $\phi = 0.7$ and use the Mat\'ern correlation kernel $\calR_M(s,s') = (1+\|s-s'\|_2/\rho)\exp(-\|s-s'\|_2/\rho)$ with known smoothness $1.5$ and unknown range parameter $\rho$, where the true $\rho$ are set as $\rho = 0.05$ or $\rho=0.1$. We set $N_i = 10$ for all locations, set $p=2$ including the intercept, with $x_{ij}\iidsim N(0,1)$. The true fixed effect coefficient is set as $\beta = (\beta_0,\beta_1)^\T = (0,1)^\T$. Denoting $N = \sum_{i=1}^n N_i$, we have training data sizes $(n,N) \in\{(200,2000), (800,8000)\}$ with two different ranges $\rho\in\{0.05, 0.1\}$, and this data generation process is repeated 200 times. 

We consider BrP (empirical Bayes) and BrP-FB with 3 different correlation kernels, one with a full-rank Matern kernel $\calR_M(s,s') = (1+\|s-s'\|_2/\rho)\exp(-\|s-s'\|_2/\rho)$ and two with a low-rank kernel \eqref{eq:corr_lowrank} with $q=49$ and $q=100$. Within the domain $[0,2]^2$, the knot locations $\{\tilde{s}_k\}_{k=1}^q$ are selected as $\{0.1, 0.4, \dots,1.6, 1.9\}^2$ for $q=49$ and $\{0.1, 0.3, \dots,1.7, 1.9\}^2$ for $q=100$, respectively. The range parameter is assumed to be unknown and we adopt the same prior specification as in Section~\ref{subsec:sim1}. For the Markov chain Monte Carlo, we run 6,000 iterations, with the first 1,000 samples discarded as burn-in and 5,000 samples are saved without thinning.

\begin{figure}
    \centering
    \includegraphics[width=\linewidth]{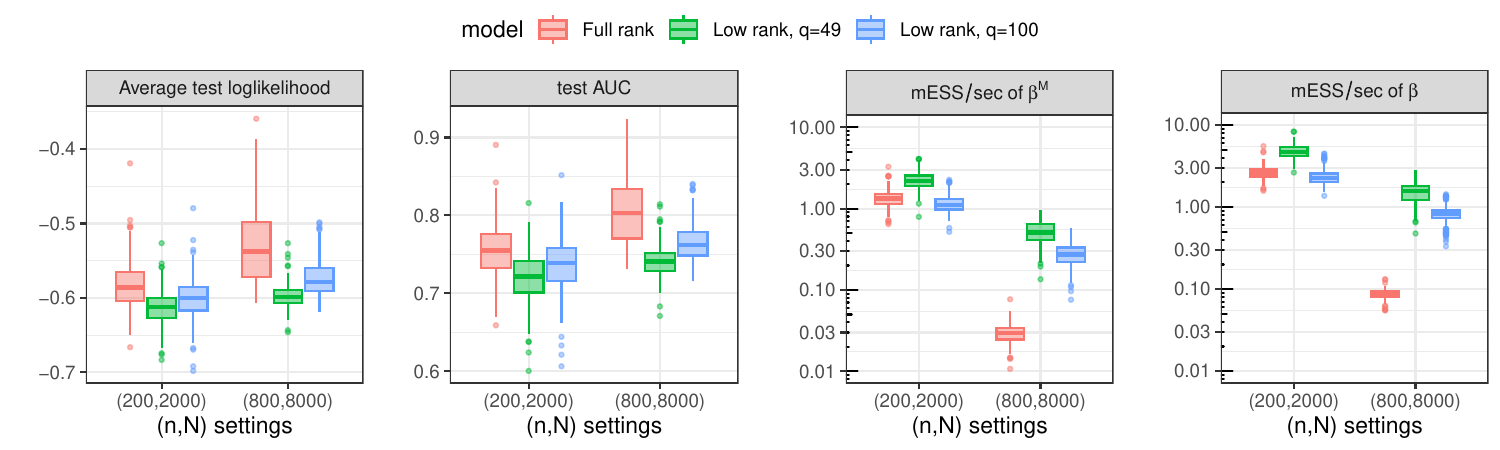}
    \caption{Boxplot summaries of scalability analysis simulation result for $\rho = 0.05$ (moderate spatial dependence) based on 200 replicated datasets. The sampling efficiency is displayed in log scale.}
    \label{fig:boxplot005}
\end{figure}

\begin{figure}
    \centering
    \includegraphics[width=\linewidth]{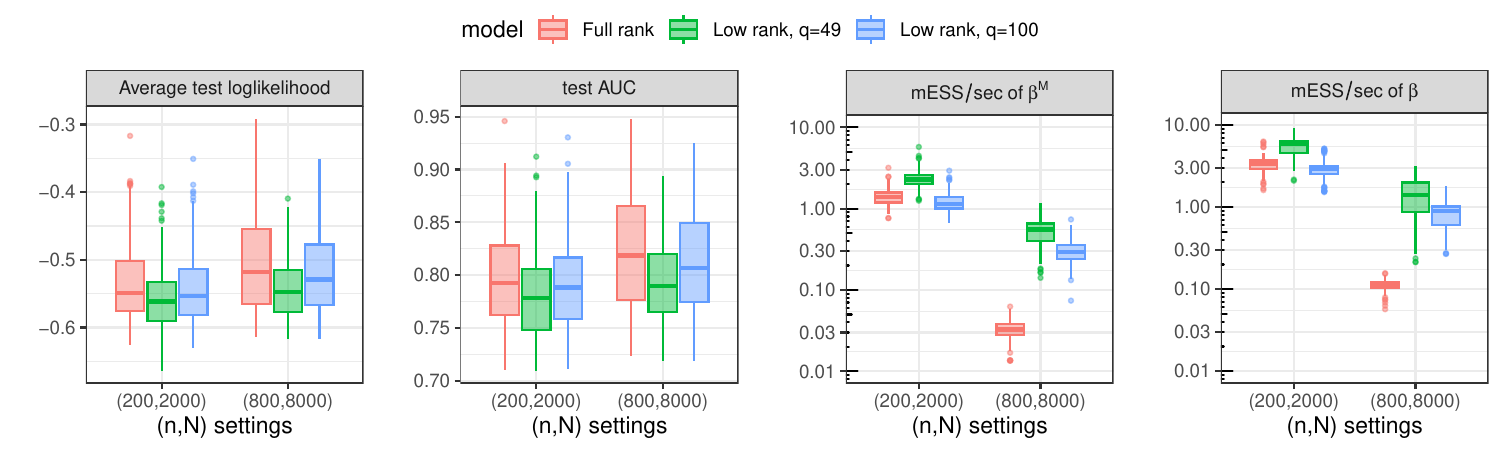}
    \caption{Boxplot summaries of scalability analysis simulation result for $\rho = 0.1$ (strong spatial dependence) based on 200 replicated datasets. The sampling efficiency is displayed in log scale.}
    \label{fig:boxplot01}
\end{figure}

The simulation results for BrP-FB are summarized in Figure~\ref{fig:boxplot005} and Figure~\ref{fig:boxplot01}. 
In terms of predictive performance based on average test log-likelihood and test AUC, the full-rank correlation kernel performs better than the low-rank ones for $\rho = 0.05$ but similarly for $\rho=0.1$ and $q=100$. This is an expected result, since low-rank methods typically yield predictions that smooth over small-scale patterns \citep{Datta2016-dx}, thus the difference is more emphasized when $\rho = 0.05$. We also compare the sampling efficiency of $\bmbetam$ and $\bm\beta$ in terms of multivariate effective sample size \citep{Vats2019-wo} divided by running time (mESS/sec). As the number of spatial locations increases from $n=200$ to $n=800$, the mESS/sec between full-rank and low-rank differs in an order of magnitude for both $\bmbetam$ and $\bm\beta$, clearly showing the computational benefits of the low-rank correlation kernel. The difference is mainly attributed to the running time difference, as mESS values not standardized by time are similar. As Table~\ref{table:1pred} suggests, the predictive performance under BrP is virtually identical to that under BrP-FB, and the sampling efficiency results for BrP (both $\bmbetam$ and $\bm\beta$) are highly similar to those for $\bm\beta$ under BrP-FB, and are omitted for better clarity.

\section{Childhood malaria prevalence in Gambia}

We illustrate the proposed model by analyzing malaria data among children in Gambia \citep{Thomson1999-ic}, The dataset, available in \texttt{R} package \texttt{geoR} \citep{Diggle1998-xw}, contains $\sum_{i=1}^n N_i = 2035$ children's malaria infection statuses from $n=65$ villages, along with covariates including age, bed net use, net treatment with insecticides, a satellite-derived measure of greenness, and the presence of a health center in the villages. The study of population-averaged associations between malaria prevalence and covariates is useful for nationwide decision-making; for example, on supplying bed nets. 
The prediction of malaria prevalence in new locations is also of substantial interest. 
The previous literature suggests residual spatial dependence in these data \citep{Diggle2002-wc,Bai2014-zc,Cattelan2018-pp}. In Supplementary Section~\ref{appendix:gambiadetails}, we provide a more detailed investigation of residual spatial dependence.

We choose the same set of variables as \citet{Cattelan2018-pp}, reproducing their results based on spatial generalized estimating equations, and comparing with bridge process (empirical Bayes for $\phi$) and Gaussian process random effect logistic models with an exponential correlation kernel. 
We use the same priors as in the simulation study \citep{Gelman2006-de,Gelman2008-hm} except for the range parameter $\rho$. 
For the prior on the spatial range $\rho$ in the exponential correlation kernel $\calR_E(s,s') = \exp(\|s-s'\|_2/\rho)$, we assign a uniform prior $\rho\sim \mathrm{Unif}(0.01,100)$ for the bridge and Gaussian process random effects models. We run three Markov chains with 11,000 iterations, where the first 1,000 samples are discarded as burn-in and 10,000 samples are saved without thinning. The wall-clock running time is about 2 mins for the bridge process random effect model with empirical Bayes approach, including the time for finding $\hat\phi$ which took approximately 1 min, and about 1.5 mins for the fully Bayesian Gaussian process random effect model under the Apple M1 3.20GHz CPU environment.

\begin{table}
\footnotesize
    \caption{Comparison of population-averaged and site-specific estimates across three methods for the Gambia malaria data. Parenthesis corresponds to 95\% confidence (credible) intervals.}
    \label{table:2}
    \centering
    \begin{tabular}{l rc rc rc rc}
    \toprule
      & \multicolumn{4}{c}{Population-average estimate ($\hat\beta_j^{\textsc{m}}$)}  & \multicolumn{4}{c}{Site-specific estimate ($\hat\beta_j$)}\\
     \cmidrule(lr){2-5}\cmidrule(lr){6-9}
    Variable & \multicolumn{2}{c}{Spatial GEE} & \multicolumn{2}{c}{Bridge process RE} & \multicolumn{2}{c}{Bridge process RE} & \multicolumn{2}{c}{Gaussian process RE}\\
     \midrule
Intercept & 6.95 & \footnotesize{(0.39, 13.52)} & 2.09 & \footnotesize{(-2.74, 6.93)} & 2.34 & \footnotesize{(-3.06, 7.74)}  & 2.43 & \footnotesize{(-3.07, 7.84)} \\
Age (years) & 0.22 & \footnotesize{(0.14, 0.30)} & 0.22 & \footnotesize{(0.14, 0.30)} & 0.24 & \footnotesize{(0.16, 0.33)} & 0.24 & \footnotesize{(0.16, 0.33)} \\
Net-use  & -0.39 & \footnotesize{(-0.66, -0.11)} & -0.33 & \footnotesize{(-0.61, -0.05)} & -0.37 & \footnotesize{(-0.68, -0.06)}  & -0.36 & \footnotesize{(-0.67, -0.06)} \\
Treated net &  -0.32 & \footnotesize{(-0.65, 0.01)} & -0.32 & \footnotesize{(-0.67, 0.02)} & -0.36 & \footnotesize{(-0.75, 0.02)}  & -0.36 & \footnotesize{(-0.75, 0.03)} \\
Green & -0.36 & \footnotesize{(-0.63, -0.08)} & -0.12 & \footnotesize{(-0.30, 0.07)} & -0.13 & \footnotesize{(-0.34, 0.07)} &  -0.13 & \footnotesize{(-0.34, 0.08)} \\
Green$^2 \times 10^2$ & 0.40 & \footnotesize{(0.11, 0.69)}  & 0.13 & \footnotesize{(-0.06, 0.32)} & 0.15 & \footnotesize{(-0.07, 0.36)} & 0.14 & \footnotesize{(-0.08, 0.36)} \\
Health center & -0.26 & \footnotesize{(-0.60, 0.08)} & -0.26 & \footnotesize{(-0.63, 0.11)} & -0.29 & \footnotesize{(-0.71, 0.12)} &  -0.30 & \footnotesize{(-0.71, 0.12)} \\
    \bottomrule
    \end{tabular}
    \vspace{-3mm}
\begin{flushleft} 
{\footnotesize 
    Spatial GEE, spatial generalized estimating equation method of \citet{Cattelan2018-pp}; RE, random effects}
\end{flushleft}
\end{table}

The results are summarized in Table~\ref{table:2}. Trace plots of the bridge process and Gaussian process random effect models are shown in Supplementary Figure~\ref{fig:mcmc_bridge}, both showing excellent convergence. The site-specific estimates using the bridge and Gaussian process random effect models give almost identical results. From the bridge process model, the empirical Bayes estimate of the attenuation factor is $\hat\phi = 0.895$. 
The spatial generalized estimating equation and the proposed method also yield similar population-averaged estimates, with slight differences in magnitude for some variables. For example, using bed net is associated with approximately $32.3\%$ reduction in the population-averaged odds of malaria infection under the Spatial GEE ($e^{-0.39}=0.677$), and a 28.1\% reduction in the population-averaged odds of malaria infection under the proposed bridge process random effects model ($e^{-0.33}=0.719$), holding other variables constant.

In addition to the estimated population-averaged and site-specific coefficients, we have a posterior mean $\hat\rho = 33.3 
 (11.2, 32.5)$ and posterior mode $\hat\rho^{\mathrm{mode}} = 10.9$ for the bridge process, and posterior mean $\hat\rho = 29.9 
 (13.4, 40.1)$ and posterior mode $\hat\rho^{\mathrm{mode}} = 13.5$ for the Gaussian process random effect models, where parentheses correspond to the posterior interquartile range. 
We assess predictive performance using the widely applicable information criterion (WAIC) \citep{Gelman2014-ga} and leave-one-out cross-validation estimate using Pareto smoothed importance sampling (PSIS-LOO) \citep{Vehtari2024-qa}, both conditional on
random effects. The WAIC is 2326.4 (standard error 39.9) for the bridge process random effect model and 2326.3 (standard error 40.0) for the Gaussian process random effect model, and the PSIS-LOO is 2326.6 (standard error 39.9) and 2326.4 (standard error 40.0) for the Gaussian process random effect model. These suggest that predictive performances are virtually identical between the proposed model and the Gaussian process random effect model, and the bridge process model remains appealing due to its dual interpretability. 

\section{Discussion}

Although we have focused on spatial settings, the proposed bridge process can also be applicable to modeling longitudinal binary data with time-correlated random effects \citep{Parzen2011-bi}. For example, for discrete-time indices $t$ and $t'$, one can employ an AR(1)-type correlation kernel $\calR(t,t') = \varrho^{|t-t'|}$ for some $\varrho\in(-1,1)$. For such a case, the corresponding posterior inference algorithm can be appropriately adjusted to accommodate multiple realizations of bridge processes with a common parameter $(\phi,\varrho)$. 
Furthermore, we anticipate that recent advances in fast approximate Bayesian methods for non-Gaussian latent models \citep{Cabral2024-jj} could also be integrated with the bridge process, leveraging the mixture representation in Theorem~\ref{thm:unimixture}.

\section*{Acknowledgement}
This research was partially supported by the National Institutes of Health (grant ID R01ES035625), by the European Research Council under the European Union’s Horizon 2020 research and innovation programme (grant agreement No 856506), by the National Science Foundation (NSF IIS-2426762), and by the Office of Naval Research (N00014-24-1-2626).
The authors thank anonymous reviewers, Amy Herring, and Georgia Papadogeorgou for helpful suggestions and discussions.

\section*{Software}

Code to reproduce the analyses is available at \url{https://github.com/changwoo-lee/spbridge}.

\bibliography{ref.bib}
\bibliographystyle{apalike}
\newpage

\begin{center}
\LARGE{Appendices}
\end{center}

\setcounter{section}{0}
\setcounter{equation}{0}
\setcounter{table}{0}
\setcounter{proposition}{0}
\renewcommand{\theequation}{A.\arabic{equation}}
\renewcommand{\thesection}{A.\arabic{section}} 
\renewcommand{\theproposition}{A.\arabic{proposition}} 
\renewcommand{\thetheorem}{A.\arabic{theorem}} 
\renewcommand{\thecorollary}{A.\arabic{corollary}} 
\renewcommand{\thelemma}{A.\arabic{lemma}} 
\renewcommand{\thetable}{A.\arabic{table}} 
\renewcommand{\thefigure}{A.\arabic{figure}}

Section A.1 contains proofs of the statements presented in the main article. Section A.2 describes the posterior inference algorithm and associated computational strategies. Section A.3 contains details of the multivariate bridge distribution, including a mixture representation.
Section A.4 includes an additional figure corresponding to the analysis of the Gambia malaria data.

\section{Proofs}
\label{appendix:proof}

\begin{proof}[Proof of Proposition~\ref{prop:varsumy}] Let $\eta_i = \bfx_i^\T\bm\beta$.  Denote the marginal mean be $E(Y_{ij}\mid \bfx_i) = \int_{-\infty}^\infty \logit^{-1}(\eta_i + u_i)  \rmd u_i=  \logit^{-1}(\phi\eta_i)= p_i$. We first show $\phi = E(\var(Y_{ij}\mid \bfx_i, u_i ))/\var(Y_{ij}\mid \bfx_i)$. We have $E(\var(Y_{ij}\mid \bfx_i, u_i )) = E[\{\logit^{-1}(\eta_i + u_i)\{ 1-\logit^{-1}(\eta_i + u_i)\}] = p_i - \int_{-\infty}^\infty \frac{(e^{u_i+\eta_i})^2}{(1+e^{u_i+\eta_i})^2} \pibr(u_i;\phi) \rmd u_i$, and the second term is 
\begin{align*}
\int_{-\infty}^\infty \frac{(e^{u_i+\eta_i})^2}{(1+e^{u_i+\eta_i})^2} \pibr(u_i;\phi) \rmd u_i &= \int_{-\infty}^\infty \frac{e^{u_i+\eta_i}(1+e^{u_i+\eta_i})-e^{u_i+\eta_i}}{(1+e^{u_i+\eta_i})^2} \pibr(u_i;\phi) \rmd u_i\\
&= \frac{e^{\phi \eta_i}}{1+e^{\phi\eta_i}} - \frac{\phi e^{\phi\eta_i}}{(1+e^{\phi\eta_i})^2}\\
&= p_i - \phi p_i (1-p_i)
\end{align*}
where in the second equality, we used 
\[
\int_{-\infty}^\infty \frac{e^{u+\eta}}{1+e^{u+\eta}} \pibr(u;\phi) \rmd u = \frac{e^{\phi \eta}}{1+e^{\phi\eta}}, \quad \int_{-\infty}^\infty \frac{e^{u+\eta}}{(1+e^{u+\eta})^2} \pibr(u;\phi) \rmd u = \frac{\phi e^{\phi\eta}}{(1+e^{\phi\eta})^2}
\]
which follows from \cite{Wang2003-ih} equations (2.1) and (2.2). Therefore, $E(\var(Y_{ij}\mid \bfx_i, u_i )) = \phi p_i (1-p_i) = \phi \var(Y_{ij}\mid \bfx_i)$. Moreover, 
\begin{align*}
\textstyle{
\var\left(\sum_{j=1}^{n_i} Y_{ij} \mid \bfx_i \right)} &= \textstyle{\var\left(E\left(\sum_{j=1}^{n_i} Y_{ij} \mid \bfx_i, u_i\right) \right) +  E\left(\var\left(\sum_{j=1}^{n_i} Y_{ij} \mid \bfx_i, u_i\right) \right)}\\
&= \var\left(n_iE\left( Y_{ij} \mid \bfx_i, u_i\right) \right) +  E\left(n_i\var\left( Y_{ij} \mid \bfx_i, u_i\right) \right)\\
&= n_i^2\var\left(E\left( Y_{ij} \mid \bfx_i, u_i\right) \right) +  n_iE\left(\var\left( Y_{ij} \mid \bfx_i, u_i\right) \right)\\
&= n_i^2 (1-\phi) \var(Y_{ij}\mid\bfx_i) +  n_i \phi \var(Y_{ij}\mid\bfx_i)\\
&= n_i \var(Y_{ij}\mid \bfx_i)\{1 + (n_i -1)(1-\phi) \}.
\end{align*}
\end{proof}

\begin{proof}[Proof of Theorem~\ref{thm:unimixture}]
The existence of a mixing distribution $\pim(\lambda; \phi)$ was shown in Proposition 1 of \citet{Wang2003-ih}, but its form was not identified.  
The characteristic function of the bridge distribution with parameter $\phi$ is $E(e^{itu}) = \sinh(\pi t)/\{\phi \sinh(\pi t / \phi)\}$ \citep{Wang2003-ih}. 
Following a similar argument as \citet{West1987-gt}, we have
\begin{align*}
E(e^{itu}) = \frac{\sinh(\pi t)}{\phi \sinh(\pi t / \phi)} &= \int_{-\infty}^\infty \int_0^\infty \exp(-itu)\mathrm{N}_1(u; 0, \lambda)\pim(\lambda; \phi) \rmd \lambda  \rmd u\\
& = \int_0^\infty \exp(-t^2\lambda /2) \pim(\lambda; \phi) \rmd \lambda =E\{\exp(-t^2\lambda/2)\}
\end{align*} 
Plugging in $t = (2s)^{1/2}$, which yields a Laplace transformation of $\lambda$,
\[
E\{\exp(-s\lambda)\} = \frac{\sinh\{\pi (2s)^{1/2}\}}{\phi \sinh\{\pi (2s)^{1/2} / \phi\}} = \prod_{k=1}^{\infty} \left\{\frac{1+2s/k^2}{1+2s/(\phi^2k^2)}\right\}, \quad s>0,
\]
where the last equation follows from the Weierstrass factorization theorem. Let $Z_k = A_kB_k$, which is a mixture of the standard exponential distribution with weight $1-\phi^2$ and a point mass at 0 with weight $\phi^2$. Recognizing that the Laplace transformation of $Z_k$ is $\phi^2 + (1-\phi^2)(1+s)^{-1} = (1+\phi^2s)/(1+s)$, the proof is completed by scaling $Z_k$ by $2/(\phi^2k^2)$ for $k\in \bbN$, respectively, and using the Laplace transformation convolution theorem.
\end{proof}

\begin{proof}[Proof of Proposition~\ref{prop:corr}]
Let $u = (u_1,\dots,u_n)$ be a finite realization from the bridge process where $\bfR$ is a matrix with $(i,j)$th element $\calR(s_i,s_j)$. By the law of total covariance, for any $i,j\in\{1,\dots,n\}$,
\begin{align*}
\cov(u_i,u_j) &= E\{\cov(u_i,u_j \mid \lambda)\} + \cov\{E(u_i \mid \lambda), E(u_j \mid \lambda)\} = E(\lambda )R_{ij}\\
& = \sum_{k=1}^\infty\frac{2(1-\phi^2)}{\phi^2 k^2}R_{ij} = \frac{\pi^2}{3}\left(\frac{1}{\phi^2}-1\right)R_{ij}
\end{align*}
which yields $\corr(u_i,u_j) = R_{ij}$ since $\var(u_i) = \pi^2/3(\phi^{-2}-1)$ for all $i=1,\dots,n$. 
\end{proof}

\begin{proof}[Proof of Proposition~\ref{prop:characterization}]
%\begin{proof}
    \citet{Kano1994-zi} showed that the family of elliptical distributions is Kolmogorov consistent if and only if it admits a scale mixture of normals representation, where the scale mixing distribution does not depend on the dimensionality. The proof is completed by the fact that the bridge-distributed marginal distribution uniquely determines a mixing distribution $\pim(\phi)$ that does not depend on the dimension, which is from the uniqueness of Laplace transformation.
\end{proof}

\begin{proof}[Proof of Proposition~\ref{prop:latentrep}] 
Although the latent variable representation is well known in the literature \citep{Holmes2006-np}, we reproduce the proof for completeness.
Let $\sigma(x) = \logit^{-1}(x) = 1/(1+e^{-x})$ and $\sigma'(x)$ be its derivative. Starting from conditional probability
\begin{align*}
\pr\{Y_{ij} = y_{ij}\mid x_{ij}, u(s_i)\} &=  \sigma\{x_{ij}^\T\beta + u(s_i)\}^{y_{ij}}[1-\sigma\{x_{ij}^\T\beta + u(s_i)\}]^{1-y_{ij}},
\end{align*}
let $Z_{ij}$ be independent logistic random variables with location $x_{ij}^\T\beta+u(s_i)$ and unit scale. Hence, $Z_{ij} = x_{ij}^\T \beta + u(s_i) + \epsilon_{ij}$, where $\epsilon_{ij}$ independently follows the standard logistic distribution. Then, since $\sigma$ is the cumulative distribution function of the standard logistic distribution, we have
\begin{align*}
    p\{Y_{ij}, Z_{ij}\mid x_{ij}, u(s_i)\} &= \{1(Z_{ij} >0)1(Y_{ij} = 1) + 1(Z_{ij} \le 0)1(Y_{ij} = 0)\}
\sigma'(\epsilon_{ij}) 
\end{align*}
which explains the latent variable representation. 

Now we show that $u(s_i) + \epsilon_{ij}$ follows a logistic distribution with scale $1/\phi$.  The product of characteristic functions of standard logistic $E(e^{it\epsilon})$ and bridge distribution $E(e^{itu})$ becomes
\[
\frac{\pi t}{\sinh(\pi t)}  \frac{\sinh(\pi t)}{\phi \sinh(\pi t / \phi)} = \frac{\pi t / \phi}{\sinh(\pi t/ \phi)}
\]
Since the right-hand side is the characteristic function of the logistic distribution with scale $1/\phi$, by the convolution theorem, this completes the proof. 
\end{proof}

\section{Details of posterior inference algorithm}
\label{appendix:post}
We first introduce the notation. Write $N = \sum_{i=1}^n N_i$, let $\bfX_i = [\bfx_{i1}^\T,\dots,\bfx_{iN_i}^\T]^\T$ be a $N_i\times p$ matrix of predictors at the $i$th location, and $\bfX = [\bfX_{1}^\T,\dots,\bfX_{n}^\T]^\T$ be a $N\times p$ fixed effects design matrix. Also, let $\bfZ = \mathrm{blockdiag}(\bm{1}_{N_1},\dots,\bm{1}_{N_n})$ be a $N\times n$ design matrix for random effects, so that the linear predictor vector becomes $\bfX\bm\beta + \bfZ\bfu_{1:n}$. 
Similarly, let $\bfy_i = (y_{i1},\dots,y_{iN_i})$ be binary responses at  location $i$ and $\bfy = (\bfy_1^\T,\dots,\bfy_n^\T)^\T$ be a response vector of length $N$. Let $\bm\Omega_i = \diag(\omega_{i1},\dots,\omega_{iN_i})$ be a $N_i\times N_i$ diagonal matrix, $\bm\Omega = \mathrm{blockdiag}(\bm\Omega_1,\dots,\bm\Omega_n)$ be a $N\times N$ diagonal matrix, and $\bm\Omega_{nn} = \bfZ^\T\bm\Omega \bfZ$ which is an $n\times n$ diagonal matrix with elements $\sum_{j=1}^{N_i}\omega_{ij}$ for $i=1,\dots,n$. Let $\bfR_\rho = [\calR(s_i,s_j)]_{i,j}$ be an $n\times n$ correlation matrix with parameter $\rho$. 

The Algorithm~\ref{alg:gibbs_fullybayes} describes a one cycle of a partially collapsed Gibbs sampler for a fully Bayesian approach. We first describe the details of the algorithm for a fully Bayesian approach, and specialize it to the  Algorithm for an empirical Bayes approach.

\begin{algorithm}
\small
\renewcommand{\thealgocf}{A1}
\caption{One cycle of a partially collapsed Gibbs sampler.}
 {[1]} Sample $\bm\beta \sim [\bm\beta \mid \bm\omega^{\mathrm{(old)}}, \lambda^{\mathrm{(old)}}, \rho^{\mathrm{(old)}}]$ from multivariate normal, where $u_{1:n}$ is collapsed out. \\
 
 {[2]} Jointly sample $(\phi,\rho, \lambda, \bfu_{1:n})\sim [\phi,\lambda, \bfu_{1:n} \mid \bm\omega^{\mathrm{(old)}}, \bm\beta]$ in two steps: \\
 \quad {[2(i)]} Sample $(\phi,\rho,\lambda) \sim [\phi,\rho,\lambda\mid \bm\omega^{\mathrm{(old)}}, \bm\beta]$ using particle marginal Metropolis–Hastings,\\
 \quad {[2(ii)]} Sample $\bfu_{1:n}\sim [\bfu_{1:n}\mid \phi,\lambda,\bm\omega^{\mathrm{(old)}}, \bm\beta]$ from multivariate normal. \\
 {[3]} Sample $\omega_{ij} \sim [\omega_{ij} \mid \bm\beta, \bfu_{1:n}]$ from P\'olya-Gamma, independently for all $i,j$.\\
\label{alg:gibbs_fullybayes}
\end{algorithm}

\subsection{Step 1}
Conditional on $\bf\Omega$, the likelihood is proportional to $\mathrm{N}_{N}\{\bm\Omega^{-1}(\bfy-0.5 \bm{1}_N); \bfX\bm\beta + \bfZ\bfu_{1:n}, \bm\Omega^{-1}\}$. Thus, integrating out $\bfu_{1:n} \sim \mathrm{N}_n(\bm{0}_n, \lambda \bfR_\rho)$, we have
\[
p(\bm\beta \mid \bm\omega, \lambda, \rho) \propto  \mathrm{N}_{N}\{\bm\Omega^{-1}(\bfy-0.5 \bm{1}_N); \bfX\bm\beta, \bm\Omega^{-1} + \lambda \bfZ \bfR_\rho \bfZ^\T\}\times N_p(\bm\beta; \bm\mu_\beta, \bm\Sigma_\beta),
\]
which yields $(\bm\beta \mid \bm\omega, \lambda, \rho) \sim \mathrm{N}_p(\bfQ_{1}^{-1}\bfb_1, \bfQ_1^{-1})$ with $\bfQ_{1} = \bfX^\T (\bm\Omega^{-1} + \lambda \bfZ \bfR_\rho \bfZ^\T)^{-1} \bfX + \bm\Sigma_\beta^{-1}$ and
$\bfb_1 = \bm\Sigma_\beta^{-1}\bm\mu_\beta + \bfX^\T (\bm\Omega^{-1} + \lambda \bfZ \bfR_{\rho} \bfZ^\T)^{-1}\bm\Omega^{-1} (\bfy - 0.5 \bm{1}_N)$. 
The inversion of $N\times N$ matrix $(\bm\Omega^{-1} + \lambda \bfZ \bfR_\rho \bfZ^\T)$ can be done efficiently using the Woodbury formula 
\begin{equation}
\label{eq:step1_woodbury}
(\bm\Omega^{-1} + \lambda \bfZ \bfR_\rho \bfZ^\T)^{-1} = \bm\Omega - \bm\Omega \bfZ(\lambda^{-1}\bfR_\rho^{-1} + \bm\Omega_{nn})^{-1}\bfZ^\T\bm\Omega     
\end{equation}
which involves the inversion of an $n\times n$ matrix instead. 

Suppose an independent normal scale mixture prior is used for $\bm\beta$, say independent $t$ priors with degrees of freedom $\nu$, mean $\mu_\beta$ and scale $\sigma_{\beta}$. This corresponds to $\beta_k\mid \gamma_k \sim \mathrm{N}(0, \gamma_k)$ and $1/\gamma_k\sim \mathrm{Ga}(\nu/2, \nu\sigma_\beta^2/2)$, independently for $k=1,\dots,p$, and one can add an additional sampling step $1/\gamma_k\sim \mathrm{Ga}(\nu/2+1/2, \nu \sigma_\beta^2/2+\beta_k^2/2)$ in the Algorithm.

\subsection{Step 2}
We first derive Step 2(ii), the full conditional distribution of $\bfu_{1:n}$. Given $\bm\Omega$ and $\bm\beta$, the likelihood $\mathrm{N}_{N}\left\{\bfZ\bfu_{1:n}; \bm\Omega^{-1}(\bfy-0.5 \bm{1}_N) - \bfX\bm\beta, \bm\Omega^{-1}\right\}$ can be recognized as a normal density in terms of $\bfu_{1:n}$ since $Z$ has full column rank. This yields the full conditional distribution
\begin{equation}
\label{eq:step2ii}
p(\bfu_{1:n} \mid  \lambda, \omega, \bm\beta, \rho) \propto \mathrm{N}_n[\bfu_{1:n}; \bm\Omega_{nn}^{-1} \bfZ^\T\bm\Omega  \{\bm\Omega^{-1}(\bfy-0.5 \bm{1}_N) - \bfX\bm\beta\}, \bm\Omega_{nn}^{-1}]\times \mathrm{N}_n(\bfu_{1:n}; \bm{0}, \lambda \bfR_\rho)   
\end{equation}
so $(\bfu_{1:n} \mid \lambda, \omega, \bm\beta, \rho) \sim \mathrm{N}_n(\bfQ_{2}^{-1}\bfb_2, \bfQ_2^{-1})$ with
$\bfQ_2 = \bm\Omega_{nn} + \lambda^{-1}\bfR_\rho^{-1}$ and $\bfb_2= \bfZ^\T \{(\bfy-0.5\bm{1}_N) - \bm\Omega \bfX \bm\beta \}$. 
For Step 2(i), from the expression \eqref{eq:step2ii}, we can obtain a collapsed likelihood with the $\bfu_{1:n}$ marginalized out. Thus, the target distribution of $\lambda$, $\phi$, and $\rho$ is given as 
\begin{equation}
\label{eq:step2i}
    p(\lambda, \phi, \rho\mid \bm\omega, \bm\beta) \propto \underbrace{\mathrm{N}_n[\bm\Omega_{nn}^{-1} Z^\T\bm\Omega  \{\bm\Omega^{-1}(\bfy-0.5 \bm{1}_N) - \bfX\bm\beta\}; \bm{0}, \bm\Omega_{nn}^{-1} + \lambda \bfR_\rho]}_{ \calL(\lambda, \rho)}  \pim(\lambda; \phi)p(\phi)p(\rho)  
\end{equation}
and when $\lambda$ is integrated out, we have $ p(\phi,\rho \mid \bm\omega, \bm\beta) \propto \int \calL(\lambda,\rho)\pim(\lambda; \phi) \rmd\lambda \times  p(\phi)p(\rho)$. 
To jointly draw $\lambda$, $\rho$ and $\phi$ from \eqref{eq:step2i}, we use a particle marginal Metropolis-Hastings sampler \citep{Andrieu2010-tm}. 
Let $L$ be the number of particles and $\lambda^{(1)},\dots,\lambda^{(L)}$ be the current set of particles. Then Step 2(i) proceeds as: (A) Draw candidate $(\phi^\star, \rho^\star) \sim q(\phi^\star, \rho^\star\mid \phi, \rho)$ with some proposal distribution. (B) Draw new particles independently from $\lambda^{\star (1)},\dots,\lambda^{\star (L)}\sim \pim(\lambda; \phi^\star)$. (C) Draw a candidate $\lambda^\star$ among new particles $\lambda^{\star (1)},\dots,\lambda^{\star (L)}$ with probability proportional to $\calL(\lambda^{\star(l)}, \rho^\star)$, $l=1,\dots,L$. (D) Accept $(\phi^\star,\rho^\star,\lambda^\star)$ and set of particles $\{\lambda^{\star(l)}\}_{l=1}^L$ with probability 
\begin{equation*}
\min\left\{1, \frac{ \{\sum_{l=1}^L \calL(\lambda^{\star(l)}, \rho^\star)\}p(\phi^\star)p(\rho^\star) }{\{\sum_{l=1}^L \calL(\lambda^{(l)}, \rho^\star)\}p(\phi)p(\rho)}\times \frac{q(\phi, \rho \mid \phi^\star, \rho^\star)}{q(\phi^\star, \rho^\star \mid \phi, \rho)}\right\}
\end{equation*}
otherwise keep $(\phi, \rho, \lambda)$ and current set of particles $\{\lambda^{(l)}\}_{l=1}^L$. 

In fully Bayesian data analysis examples, we set $L=20$ and use coordinatewise logit transform to map parameters $(\phi,\rho)\in(0,1)\times (0.001, 0.3)$ to $\bbR^2$, and utilize Metropolis-Hastings proposal adaptation settings as in \citet{Haario2001-ms} to determine proposal distribution $q(\phi^\star,\rho^\star\mid \phi,\rho)$. 

\begin{remark} For the joint sampling of $(\lambda, \phi, \rho)$, the particle marginal Metropolis-Hastings sampler \citep{Andrieu2010-tm} is suitable for this setting for the following reasons. First, the parameter $\phi$ only depends on $\lambda$, and conditional sampling between $[\phi|\lambda, -]$ and $[\lambda |\phi, -]$ leads to very high autocorrelation of $\phi$ and $\lambda$, and joint sampling of $\phi$ and $\lambda$ is highly desired. Next, from the joint density $p(\phi)\pim(\lambda;\phi)$, the mixing density $\pim(\lambda;\phi)$ is represented as a alternating series (see Proposition~\ref{prop:mixingdensity}) and may unstable to evaluate, but it is easy to sample $\lambda$ using Theorem~\ref{thm:unimixture} and thus the joint density evaluation can be replaced with Monte Carlo estimates with particles, which is unbiased. Finally, using a simple Metropolis-Hastings proposal may suffer from a low acceptance probability, where the particle marginal Metropolis-Hastings algorithm leverages multiple proposals and improves the acceptance rate. 
\end{remark}

\subsection{Step 3 and additional remarks}
Step 3 corresponds to sampling auxiliary variables $\omega_{ij} \mid \bm\beta, \bfu_{1:n}\sim \textrm{P\'olya-Gamma}\{1, \bfx_{ij}^\T \bm\beta + u(s_i)\}$ for all $i,j$, which is straightforward from equation \eqref{eq:expandedmodel1}, see \citet{Polson2013-gb}.

%When there are additional unknown parameters $\vartheta$, such as parameters of the correlation kernel,
%having prior $p(\vartheta)$, a sampling step for $\vartheta$ can be included in Step 2(i) with a suitable proposal $\vartheta^\star \sim q(\vartheta^\star \mid \vartheta)$. This is beneficial for mixing, since random effects $\bfu_{1:n}$ are integrated out in Step 2(i). For the choice of proposal distribution in Step 2(i), we transform the parameter space to the Euclidean space and apply the adaptive scheme of \citet{Haario2001-ms} based on a multivariate normal proposal. 
When it comes to the empirical Bayes approach, where $\phi$ is fixed at $\hat\phi$, there are fewer benefits of using the particle marginal Metropolis-Hastings algorithm. In this case, we suggest sampling $\rho$ given $\lambda$ (Stage 2 - Step 2 of Algorithm 1) and $\lambda$ given $\rho$ (Stage 2 - Step 3 of Algorithm 1) using the Metropolis-Hastings algorithm based on expression \eqref{eq:step2i}, where we use an adaptive Metropolis-Hastings proposal for $\rho$ and Metropolis-Hastings for $\lambda$ with prior $\pim(\lambda; \phi)$ as an independent proposal distribution, which avoids evaluation of $\pim(\lambda; \phi)$ as they cancel out in the acceptance probability.

\subsection{Details of scalable computation with low-rank dependence structure}

We describe in detail how Algorithm~\ref{alg:gibbs} becomes scalable with the correlation kernel \eqref{eq:corr_lowrank} with a low-rank structure. Denoting $\bfR_{nq} = [\calR(s_i,\tilde{s}_k)]_{i=1, k=1}^{n,q}$ as an $n\times q$ matrix and $\bfR_{qn} = \bfR_{nq}^\T$, we have a low-rank structured correlation matrix for $n$ realizations $\tilde{\bfR} = [\tilde{\calR}(s_i,s_{i'})]_{i,i'=1}^{n,n}= \bfR_{nq}\bfR_{qq}^{-1}\bfR_{qn} + \bfD_{nn}$ where $\bfD_{nn}$ is a diagonal matrix with elements $1-\bfr(s_i)^\T \bfR_{qq}^{-1}\bfr(s_{i})$ for $i=1,\dots,n$. 

First, step 1 of Algorithm 1 or Algorithm 2 involves $n\times n$ matrix inversion in equation \eqref{eq:step1_woodbury}. This can be reduced to a $q\times q$ matrix inversion problem using the Woodbury matrix identity,
\begin{align}
 (\lambda^{-1}\tilde{\bfR}^{-1} + \bm\Omega_{nn})^{-1} &= \lambda\{(\bfR_{nq}\bfR_{qq}^{-1}\bfR_{qn}+\bfD_{nn})^{-1} + \lambda \bm\Omega_{nn}\}^{-1} \nonumber  \\
 &= \lambda \bm\Delta_1 + \lambda \bm\Delta_1 \bfD_{nn}^{-1}\bfR_{nq}(\bfR_{qq}+ \bfR_{qn}\bm\Delta_2 \bfR_{nq})^{-1}\bfR_{qn} \bfD_{nn}^{-1}\Delta_1
 \label{eq:step1woodburytwice}
\end{align}
where $\bm\Delta_1 = (\lambda \bm\Omega_{nn} + \bfD_{nn}^{-1})^{-1}$, $\bm\Delta_2 = \bfD_{nn}^{-1} - \bfD_{nn}^{-1}\bm\Delta_1 \bfD_{nn}^{-1}$ are diagonal matrices. 

Next, step 2, step 3 of Algorithm 1, and step 2(i) of Algorithm 2 involve $n$-dimensional normal density evaluation with covariance $\lambda \tilde{\bfR} + \bm\Omega_{nn}^{-1}$. Its inverse and determinant can be efficiently calculated as
\begin{align*}
(\lambda \tilde{\bfR} + \bm\Omega_{nn}^{-1})^{-1} &= (\lambda \bfR_{nq}\bfR_{qq}^{-1}\bfR_{qn} + \lambda \bfD_{nn} + \bm\Omega_{nn}^{-1})^{-1}\\
&= \bm\Delta_3 - \bm\Delta_3\bfR_{nq}(\lambda^{-1} \bfR_{qq}+ \bfR_{qn}\bm\Delta_3 \bfR_{nq} )^{-1}\bfR_{qn}\bm\Delta_3,\\
|\lambda \tilde{\bfR} + \bm\Omega_{nn}^{-1}| &= |\bm\Delta_3^{-1}|\times | \lambda \bfR_{qq}^{-1}|\times  |\lambda^{-1} \bfR_{qq} + \bfR_{qn}\Delta_3 \bfR_{nq}| 
\end{align*}
where $\bm\Delta_3 = (\lambda \bfD_{nn} + \bm\Omega_{nn}^{-1})^{-1}$ is a diagonal matrix.

Finally, step 2(ii) corresponds to sampling $(\bfu_{1:n} \mid \lambda, \bm\omega, \bm\beta) \sim \mathrm{N}_n(\bfV_{2}\bfb_2, \bfV_2)$ with
$\bfV_2 = (\bm\Omega_{nn} + \lambda^{-1}\tilde{\bfR}^{-1})^{-1}$ and $\bfb_2= \bfZ^\T \{(\bfy-0.5\bm{1}_N) - \bm\Omega \bfX \bm\beta \}$. Using the expression \eqref{eq:step1woodburytwice} for $\bfV_2$, $\bfu_{1:n}$ can be sampled from the following steps: (1) Sample from $q$-dimensional multivariate normal  $\tilde{\bfu}\sim \mathrm{N}_q\{\bm{0}, (\bfR_{qq}+\bfR_{qn}\Delta_2\bfR_{nq})^{-1}\}$, (2) Sample $\epsilon\sim \mathrm{N}_n(\bm{0}, \lambda\bm\Delta_1)$, (3) Set $\bfu_{1:n} = \bfV_2\bfb_2 + \lambda^{1/2}\bm\Delta_1 \bfD_{nn}^{-1}\bfR_{nq}\tilde{\bfu} + \epsilon$, which reduces to sampling from a $q$-dimensional multivariate normal instead of an $n$-dimensional one. 

\section{Details of multivariate bridge  distribution}
\label{appendix:multibridge}
We describe the probability density function of the mixing distribution $\pim(\phi)$ and multivariate bridge distribution. We also remark on the sampling procedure for the mixing distribution. 
\begin{proposition} The density of normal variance mixing distribution of bridge distribution is 
    \begin{equation}
\pim(\lambda; \phi) = \frac{(\pi/2)^{1/2}}{ \phi^2\lambda^{3/2}}\sum_{k=1}^\infty (-1)^{k+1}C_k(\phi) \exp\left\{-\frac{\pi^2C_k(\phi)^2}{2 \phi^2 \lambda}\right\}, \quad \lambda>0,        
    \end{equation}
    where $C_k(\phi) = k- 0.5 + (-1)^k(\phi-0.5)$. 
\end{proposition}

\begin{proof}Recall that from the proof of Theorem~\ref{thm:unimixture}, the Laplace transform of $\lambda \sim \pim(\phi)$ is $E(e^{-s\lambda}) = \sinh(\pi (2s)^{1/2})/ \left\{\phi\sinh(\pi (2s)^{1/2}/\phi)\right\}$. 
    From the random variable $T_1^\phi(R_3)$ studied by \citet{Biane2001-nz} which has the Laplace transform $ E(e^{-sT_1^\phi(R_3)}) = \sinh\{ \phi (2 s)^{1/2} \}/[\phi \sinh\{( 2 s)^{1/2}\}]$, we can see that $\lambda$ is equal in distribution to $(\pi^2/\phi^2) T_1^\phi(R_3)$. Then, the proposition follows directly from the density of $T_1^{\phi}(R_3)$, which is given on page 460 of \citet{Biane2001-nz}. 
\end{proof}

Next, we present the density function of the multivariate bridge distribution. In practice, this alternating sum representation is not used in the posterior inference procedure. This is because the inference of $\phi$ is based on the collapsed conditional where the multivariate bridge distribution is further integrated in Step 2 (i) of Algorithm~\ref{alg:gibbs}, and conditioning on the mixing variable $\lambda$ is the key to preserving the conditional conjugacy of the remaining steps.

\begin{proposition}
\label{prop:mixingdensity}
For $d\ge 2$, the density function of a $d$-dimensional multivariate bridge distribution with parameters $\phi$ and positive definite $R$ can be represented as  
\begin{equation}
\pibr(u; \phi, R) = \frac{\Gamma(d/2+1/2)}{\phi^2\pi^{(d-1)/2}|R|^{1/2}}\sum_{k=1}^\infty \frac{ 
(-1)^{k+1} C_k(\phi)}{\{\pi^2 C_k(\phi)^2/\phi^2 + u^\T R^{-1}u\}^{(d+1)/2}}
\end{equation}
with  $C_k(\phi) = k- 0.5 + (-1)^k(\phi-0.5)$. 
\end{proposition}
\begin{proof}
Let $D(\phi) = \{\phi^2 2^{(d+1)/2}\pi^{(d-1)/2}|R|^{1/2}\}^{-1}$. Then
\begin{align*}
    \pibr(u; \phi, R) &=  \frac{1}{(2\pi)^{d/2}|R|^{1/2}}\int_0^{\infty}\frac{1}{\lambda^{d/2}}\exp\left(-\frac{1}{2\lambda }u^\T R^{-1}u\right)\pim(\lambda; \phi) d\lambda\\
    &= D(\phi) \int_{0}^\infty     
    \sum_{k=1}^\infty\frac{ (-1)^{k+1}C_k(\phi)}{\lambda^{(d+3)/2}} \exp\left[-\frac{1}{2\lambda}\left\{\frac{\pi^2C_k(\phi)^2}{\phi^2 }+u^\T R^{-1}u \right\}\right] d\lambda\\
    &\stackrel{(*)}{=} D(\phi) \sum_{k=1}^\infty (-1)^{k+1}C_k(\phi) \int_{0}^\infty\frac{1}{\lambda^{(d+3)/2}}\exp\left[-\frac{1}{2\lambda}\left\{\frac{\pi^2C_k(\phi)^2}{\phi^2 }+u^\T R^{-1}u \right\}\right] d\lambda\\
    &\stackrel{(**)}{=} D(\phi) \sum_{k=1}^\infty (-1)^{k+1}C_k(\phi) \frac{\Gamma(d/2+1/2)2^{(d+1)/2}}{\{\pi^2 C_k(\phi)^2/\phi^2 + u^\T R^{-1}u\}^{(d+1)/2}}\\
    &= \frac{\Gamma(d/2+1/2)}{\phi^2\pi^{(d-1)/2}|R|^{1/2}}\sum_{k=1}^\infty \frac{ 
(-1)^{k+1} C_k(\phi)}{\{\pi^2 C_k(\phi)^2/\phi^2 + u^\T R^{-1}u\}^{(d+1)/2}}
\end{align*}
where $(*)$ holds for $d\ge 2$ due to Fubini, and $(**)$ is from recognizing the integrand as an inverse gamma density with parameters $(d+1)/2$ and $\{\pi^2 C_k(\phi)^2/\phi^2 + u^\T R^{-1}u\}/2$. 
\end{proof}

Approximations can be used to sample
$\lambda\sim \pim(\phi)$. A naive method is based on finite truncation at some large level $K_1$, namely, sampling from $2\phi^{-2}\sum_{k=1}^K A_kB_k/k^2$ with $A_k\sim \mathrm{Exp}(1)$ and $B_k\sim \mathrm{Ber}(1-\phi^2)$. However, this approach leads to non-zero probability
 $\phi^{2K_1}$ of $\lambda$ being exactly zero; we have observed this in practice when $\phi$ is close to one. To avoid this issue, we use geometric random variables corresponding to the waiting time between Bernoulli successes. With some large $K_2$, say $K_2 = 100$, we first generate $K_2$ independent geometric random variables $C_1,\dots,C_{K_2}$ with success probability $1-\phi^2$. Then, we approximately sample $\lambda$ from sum of $K_2$ independent exponential distributions with scales $2/(\phi^2D_k^2)$ for $k=1,\dots,K_2$, where $D_k = \sum_{l=1}^k C_l$.

\newpage 

\section{Additional results for simulation study}
\label{appendix:simdetails}

\begin{table}[h]
\footnotesize
\caption{Comparison of estimated population-averaged effects based on 200 simulations in terms of bias, root mean squared error (RMSE), average length of 95\% confidence or credible interval (CI$_{.95}$), and coverage probabilities at nominal level 0.95 (Cover).}
    \label{table:sim_beta0_mar}
    \centering
    \begin{tabular}{cc c c c c c c c c}
    \toprule
     \multicolumn{2}{c}{Population-averaged ($\hat\beta_0^\textsc{m}$)} & \multicolumn{4}{c}{Data from bridge process} & \multicolumn{4}{c}{Data from Gaussian copula}\\
     \cmidrule(lr){3-4}\cmidrule(lr){5-6} \cmidrule(lr){7-8}\cmidrule(lr){9-10}
     & & \multicolumn{2}{c}{$\phi = 0.7$} & \multicolumn{2}{c}{$\phi = 0.9$} & \multicolumn{2}{c}{$\phi = 0.7$} & \multicolumn{2}{c}{$\phi = 0.9$} \\
      \cmidrule(lr){3-6} \cmidrule(lr){7-10}
     Setting & Method & Bias & RMSE & Bias & RMSE & Bias & RMSE & Bias & RMSE\\ 
     \midrule
    \multirow{3}{*}{\makecell{ $\rho = 0.05$ }} &  SpGEE & 0.001 & 0.237 & -0.019 & 0.165 & -0.005 & 0.227 & -0.024 & 0.147\\
    & BrP & -0.003 & 0.231& -0.019 & 0.168 & -0.003 & 0.235 & -0.023 & 0.151\\
    & BrP-FB & -0.003 & 0.245 & -0.022 & 0.169 & -0.005 & 0.237 & -0.021 & 0.135 \\
    \midrule
    \multirow{3}{*}{\makecell{ $\rho = 0.1$ }}& SpGEE & -0.013 & 0.386 & -0.018 & 0.298 & -0.011 & 0.380 & -0.034 & 0.231 \\
        & BrP & -0.001 & 0.028 & -0.024 & 0.310 & -0.002 & 0.404 & -0.032 & 0.247  \\
    & BrP-FB & -0.002 & 0.382 & -0.027 & 0.299 & 0.003 & 0.377 & -0.025 & 0.220 \\
    \bottomrule
    \end{tabular}
    \vspace{-3mm}
\begin{flushleft} 
{\footnotesize SpGEE, spatial generalized estimating equation method of \citet{Cattelan2018-pp}; BrP, proposed bridge process model; BrP-FB, bridge process model with fully Bayesian approach on $\phi$.}
\end{flushleft}
\end{table}

\begin{table}[h]
\footnotesize
\caption{Comparison of estimated site-specific effects based on 200 simulations in terms of bias, root mean squared error, average length of 95\% confidence or credible interval, and coverage probabilities at nominal level 0.95.}
    \label{table:sim_beta0_cond}
    \centering
    \begin{tabular}{cc c c c c c c c c}
    \toprule
     \multicolumn{2}{c}{Site-specific ($\hat\beta_0$)} & \multicolumn{4}{c}{Data from bridge process} & \multicolumn{4}{c}{Data from Gaussian copula}\\
     \cmidrule(lr){3-4}\cmidrule(lr){5-6} \cmidrule(lr){7-8}\cmidrule(lr){9-10}
     & & \multicolumn{2}{c}{$\phi = 0.7$} & \multicolumn{2}{c}{$\phi = 0.9$} & \multicolumn{2}{c}{$\phi = 0.7$} & \multicolumn{2}{c}{$\phi = 0.9$} \\
      \cmidrule(lr){3-6} \cmidrule(lr){7-10}
     Setting & Method & Bias & RMSE & Bias & RMSE & Bias & RMSE & Bias & RMSE\\ 
     \midrule
    \multirow{3}{*}{\makecell{ $\rho = 0.05$ }} &  GP & -0.022 & 0.412 & -0.041 & 0.274 & -0.005 & 0.338 & -0.026 & 0.171 \\
    & BrP & -0.022 & 0.401 & -0.040 & 0.278 & -0.004 & 0.337 & -0.025 &0.168 \\
    & BrP-FB & -0.021 & 0.413 & -0.040 & 0.276 & -0.007 & 0.339 & -0.026 & 0.170  \\
    \midrule
    \multirow{3}{*}{\makecell{ $\rho = 0.1$ }}& GP & -0.019 & 0.617 & -0.057 & 0.505 & 0.006 & 0.568 & -0.035 & 0.295\\
        & BrP & -0.018 & 0.616 & -0.056 & 0.497 & 0.000 & 0.560 &-0.036 &  0.275 \\
    & BrP-FB & -0.019 & 0.621 &-0.059 & 0.516& 0.008 & 0.568 & -0.033 & 0.293   \\
    \bottomrule
    \end{tabular}
    \vspace{-3mm}
\begin{flushleft} 
{\footnotesize SpGEE, spatial generalized estimating equation method of \citet{Cattelan2018-pp}; BrP, proposed bridge process model; BrP-FB, bridge process model with fully Bayesian approach on $\phi$.}
\end{flushleft}
\end{table}

\newpage

\section{Additional results for Gambia malaria data study}
\label{appendix:gambiadetails}

To better motivate the need for the spatial logistic model in analyzing the Gambia data, we analyzed the spatial dependence of the residuals through the variogram and Moran's I test \citep{Wikle2019-tq}. We use village-level Pearson residual, defined as
\[
r_i = \frac{1}{\sqrt{N_i}}\sum_{j=1}^{N_i}\frac{Y_{ij} - \hat{p}_{ij}}{\sqrt{\hat{p}_{ij}(1-\hat{p}_{ij})}}
\]
for village $i$ ($\hat{p}_{ij}$ is an estimated probability of $Y_{ij}$), described in Section 3.1. of \citet{Diggle2002-wc}. Specifically, we analyze the village-level Pearson residuals under three models: (A) the logistic random intercept model (non-spatial), (B) the bridge process random effect logistic model, and (C) the Gaussian process random effect logistic model.

\begin{figure}
    \centering    
    \includegraphics[width=\linewidth]{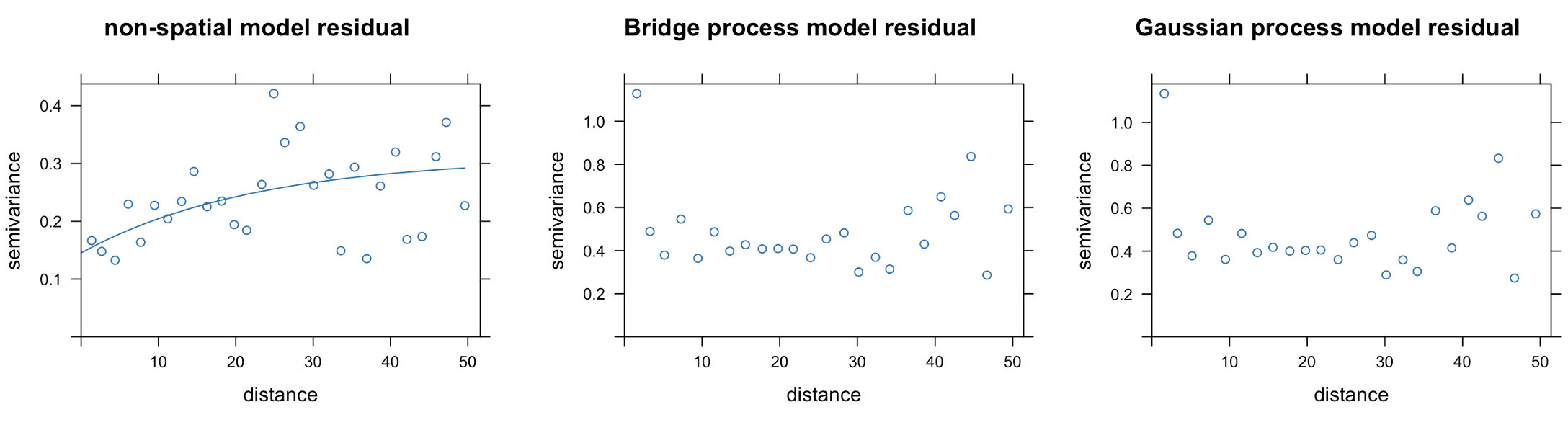}
    \caption{
Empirical variograms of village-level residuals $r_i$ for three logistic regression models (distance in km). (Left, A) Non-spatial random-intercept logistic model, which exhibits a clear increasing trend with distance. (Center, B) Bridge-process random-effect model, and (Right, C) Gaussian-process random-effect model, both of which show no evident spatial pattern. For (A), the fitted exponential variogram model is overlaid. For (B) and (C), fitting an exponential variogram model did not converge.
}
    \label{fig:variograms}
\end{figure}

The empirical variograms based on residuals from the three models are displayed in Figure~\ref{fig:variograms}. The variogram for model (A) shows a clear increasing pattern with distance, indicating substantial residual spatial dependence and motivating the need for spatial modeling. In contrast, the variograms for models (B) and (C) exhibit no discernible spatial structure, suggesting that the spatial random effects in these models successfully account for the underlying spatial dependence.

Moran’s I test using a 4-nearest-neighbor weight matrix further supports these observations. For model (A), the test yields $I = 0.334$ ($p\text{-value} < 10^{-5}$), providing strong evidence against the null hypothesis of no spatial correlation. For models (B) and (C), the results are $I = -0.193$ ($p\text{-value} = 0.987$) and $I = -0.202$ ($p\text{-value} = 0.991$), respectively, indicating no evidence of residual spatial dependence.

%The empirical variograms based on residuals from three different models are shown in Figure~\ref{fig:variograms}. The residual from (A) indicates clear spatial dependence, motivating the need for spatial modeling. In contrast, the residuals from (B) and (C) show no remaining spatial structure, suggesting that the spatial random effects in these models adequately capture the underlying spatial dependence. The Moran's I test, with a 4 nearest neighbor weight matrix, yields $I=0.334$ ( $pvalue < 10^{-5}$ ) suggests very strong evidence against the null hypothesis of no spatial correlation for (A). Meanwhile, for B and C we have $I=-0.193$ ( $pvalue = 0.987$ ) and $I=-0.202$ ( $pvalue < 0.991$ ), ...

\begin{figure}[h]
\centering
\includegraphics[width=0.9\textwidth]{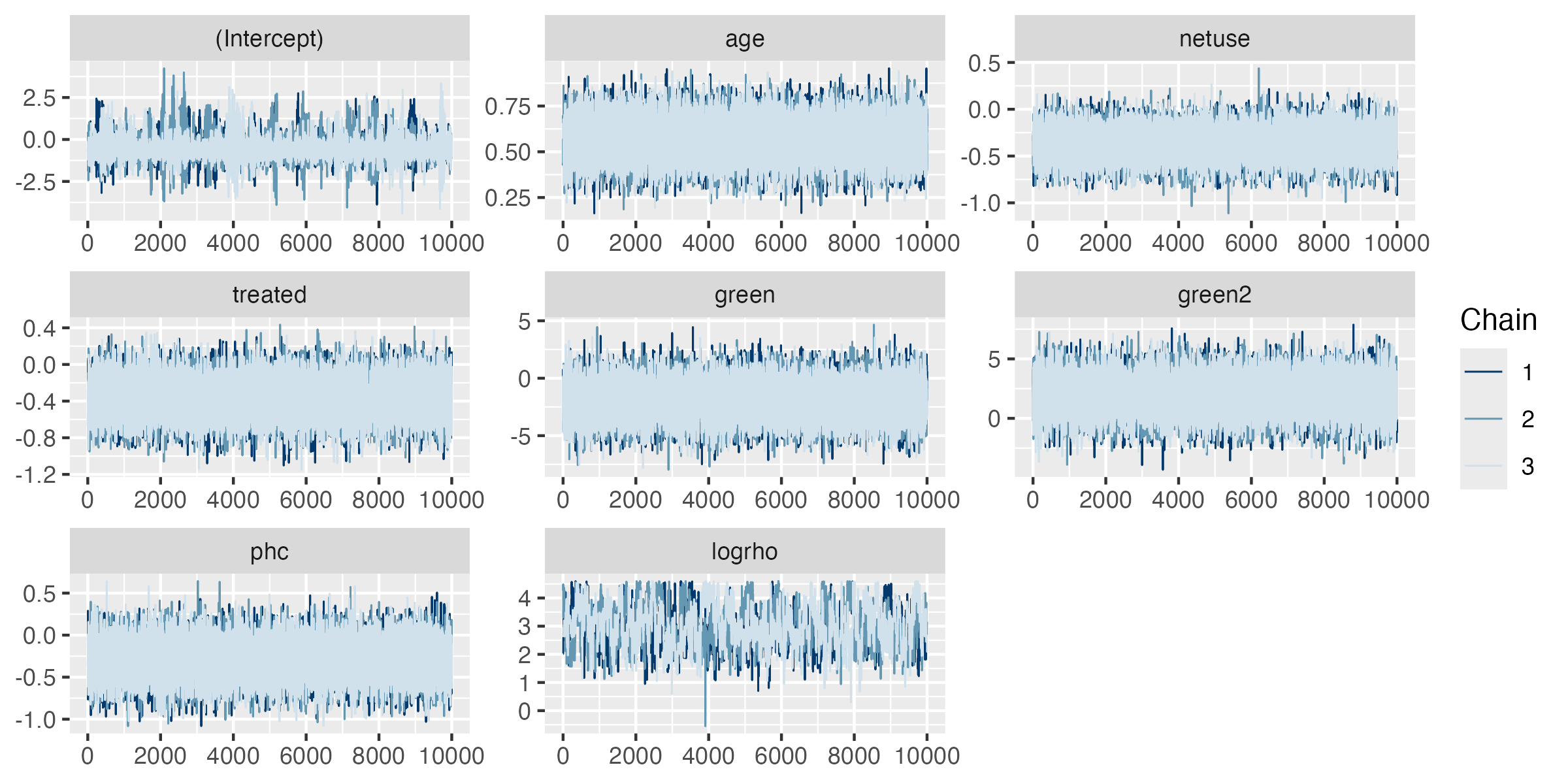}
\vspace{5mm}
\includegraphics[width=0.9\textwidth]{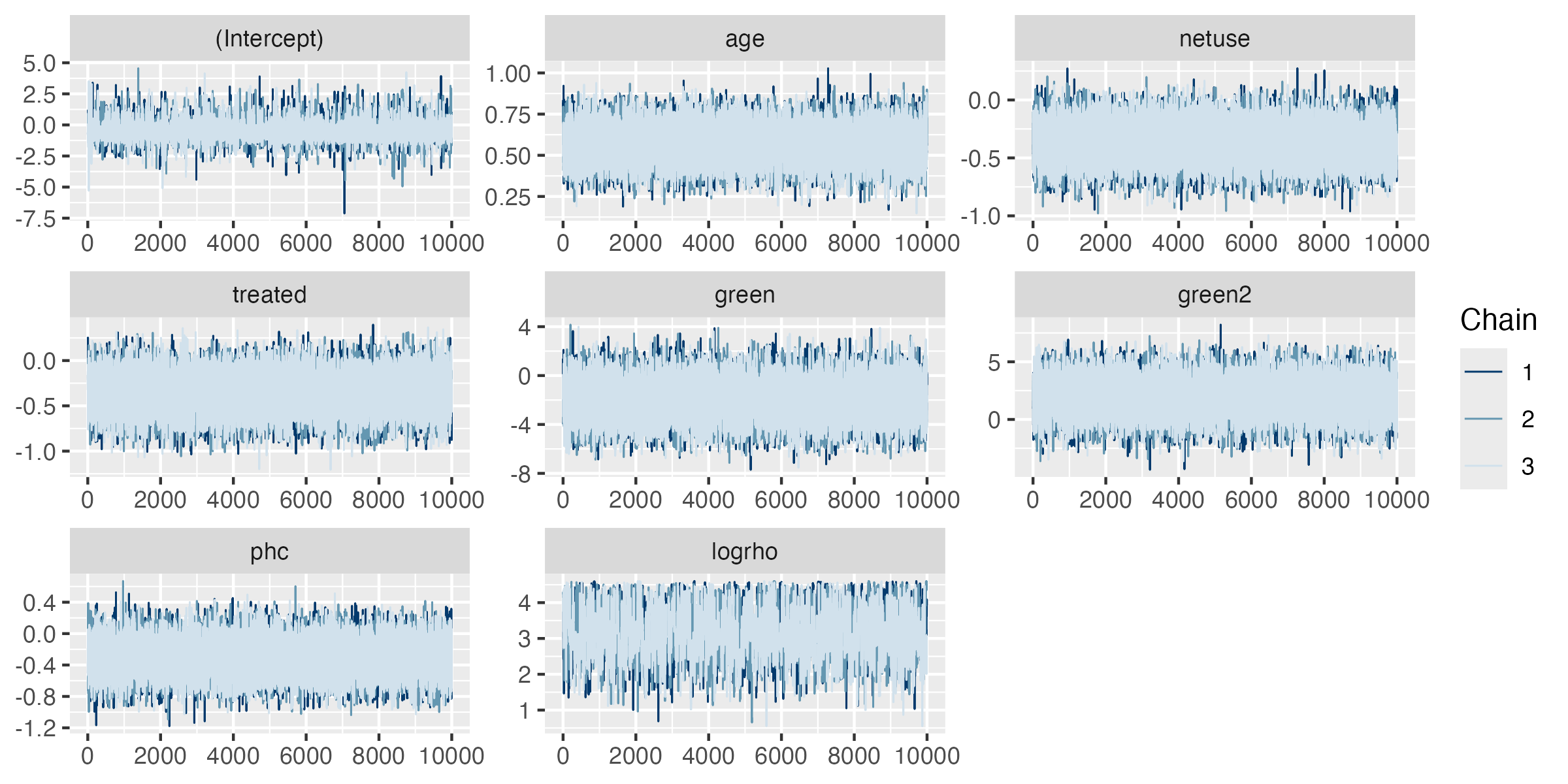}
\caption{Markov chain Monte Carlo trace plots for the Gambia childhood malaria data analysis based on the Algorithm~\ref{alg:gibbs} based on 3 different chains. (Top) Proposed bridge process random effects logistic model, (Bottom) Gaussian process random effects logistic model.}
 \label{fig:mcmc_bridge}
\end{figure}

Finally, Figure~\ref{fig:mcmc_bridge} shows MCMC trace plots of model parameters from 3 different chains, illustrating good mixing for both the bridge process random effect model and the Gaussian process random effect model.

\end{document}